\documentclass[runningheads]{llncs}
\usepackage{graphicx}

\usepackage{mathtools,dsfont}
\usepackage{psfrag,bbm,graphicx}
\usepackage{comment,balance}
\usepackage{epstopdf}
\usepackage{breakurl}
\usepackage{tikz}
\usepackage{url}
\usepackage{svg}
\usepackage{float}
\usepackage{cite, hyperref}

\usepackage{epsfig,subfig}
\usepackage{multicol}
\usepackage{multirow}
\usepackage{amsfonts, color}
\usepackage{array}
\usepackage{amssymb}
\usepackage[linesnumbered,ruled,vlined]{algorithm2e}
\usepackage{nccmath}
\usepackage{caption}
\usepackage{bm}
\usepackage{pgfplots}
\usetikzlibrary{arrows.meta, positioning, quotes}

\newtheorem{assumption}{Assumption}

\newcommand{\remove}[1]{{}}

\makeatletter
\def\endthebibliography{%
	\def\@noitemerr{\@latex@warning{Empty `thebibliography' environment}}%
	\endlist
}
\makeatother
\begin{document}
\title{Renting Edge Computing \\ Resources for Service Hosting} 
%
%
\author{Aadesh Madnaik 
\and
Sharayu Moharir
\and
Nikhil Karamchandani
}
\authorrunning{A. Madnaik et al.}
%
\institute{Indian Institute of Technology Bombay, Mumbai, India \\
\email{aadesh.madnaik@gmail.com} \\ \email{\{sharayum,nikhilk\}@ee.iitb.ac.in}}
\maketitle              
\begin{abstract}
We consider the setting where a service is hosted on a third-party edge server deployed close to the users and a cloud server at a greater distance from the users. Due to the proximity of the edge servers to the users, requests can be served at the edge with low latency. However, as the computation resources at the edge are limited, some requests must be routed to the cloud for service and incur high latency. The system's overall performance depends on the rent cost incurred to use the edge server, the latency experienced by the users, and the cost incurred to change the amount of edge computation resources rented over time. The algorithmic challenge is to determine the amount of edge computation power to rent over time. We propose a deterministic online policy and characterize its performance for adversarial and stochastic i.i.d. request arrival processes. We also characterize a fundamental bound on the performance of any deterministic online policy. Further, we compare the performance of our policy with suitably modified versions of existing policies to conclude that our policy is robust to temporal changes in the intensity of request arrivals.

\end{abstract}

\section{Introduction}
\label{sec:intro}
Software as a Service (SaaS) instances like search engines, online shopping platforms, navigation services, and Video-on-Demand services have recently gained popularity. Low latency in responding to user requests/queries is essential for most of these services. This necessitates the use of edge resources in a paradigm known as edge-computing \cite{satyanarayanan2017emergence}, i.e., storage and computation power close to the resource-constrained users, to serve user queries. Due to limited computation resources at the edge, such services are often also deployed on cloud servers which can serve requests that cannot be served at the edge, albeit with more latency given the distance between the cloud servers and the users. Ultimately, introducing edge-computing platforms facilitates low network latency coupled with higher computational capabilities. For instance, consider a scenario where a child goes missing in an urban setting \cite{shi2016edge}. While cameras are widely used for security, it is challenging to leverage the information as a whole because of privacy and data traffic issues. In the edge computing paradigm, a workaround would be to push a request to search for the child to a certain subset of devices, thereby making the process faster and more efficient than cloud computing. Several other avenues of edge computing exist in the forms of cloud offloading, AR/VR-based infotainment, autonomous robotics, Industry 4.0 and the Internet of Things (IoT). Several industry leaders offer services for edge resources, e.g., Amazon Web Services \cite{AWS}, Oracle Cloud Infrastructure \cite{OCI} and IBM with 5G technology \cite{IBM}.

This work considers a system with cloud servers and third-party-owned edge servers.  Edge resources, i.e., storage and computation power, can be rented via short-term contracts to host services. Storage resources are needed to store the code, databases, and libraries for the service and computation resources are required to compute responses to user queries. As edge servers are limited in computational capabilities, there is a cap on the number of concurrent requests that can be served at the edge \cite{tran2019costa}. The amount of edge computational resources rented for the service governs the number of user requests that can be served simultaneously at the edge. We focus on a service that is hosted both on the cloud and edge servers, and the amount of edge computational resources rented can be changed over time based on various factors, including the user request traffic and the cost of renting edge computation resources. Service providers provision for elasticity in the quantity of edge resources rented, and the clients can exploit this based on the number of request arrivals \cite{mouradian2017comprehensive}. 
The total cost incurred by the system is modelled as the sum of the rent cost incurred to use edge resources, the cost incurred due to high latency in serving requests that have to be routed to the cloud, and the switching cost incurred every time the amount of edge computation resource rented is changed \cite{zhang2018power}. The algorithmic challenge in this work is to determine the amount of edge computation resources to rent over time in the setting where the request arrival sequence is revealed causally with the goal of minimizing the overall cost incurred. 

\subsection{Our Contributions}
We propose a deterministic online policy called Better-Late-than-Never (BLTN) inspired by the RetroRenting policy proposed in \cite{narayana2021renting} and analyze its performance for adversarial and, i.i.d. stochastic request arrival patterns. In addition to this, we also characterize fundamental limits on the performance of any deterministic online policy for adversarial arrivals in terms of competitive ratio against the optimal-offline policy. Further, we compare the performance of BLTN with a suitably modified version of the widely studied Follow the Perturbed Leader (FTPL) policy \cite{mukhopadhyay2021online,bhattacharjee2020fundamental} via simulations. Our results show that while the performance of BLTN and FTPL is comparable for i.i.d. stochastic arrivals, for arrival processes with time-varying intensity, e.g., a Gilbert-Elliot-like model, BLTN significant outperforms FTPL. The key reason for this is that BLTN puts extra emphasis on recent arrival patterns of making decisions, while FTPL uses the entire request arrival history to make decisions. For all settings under consideration, the simulations demonstrate that BLTN differs little in performance from the optimal online policy despite not having information about the incoming request arrival process.

\color{black}

\subsection{Related Work} 
There has been a sharp increase in mobile application latency and bandwidth requirements, particularly when coupled with time-critical domains such as autonomous robotics and the Internet of Things (IoT).  These changes have ushered in the advent of the edge computing paradigm away from the conventional remote servers, as discussed in the surveys \cite{Puliafito:2019, 
luo2021surveyresource, abbas2018surveymobile}. The surveys alongside several academic works elaborate on and model the dynamics of such systems. We briefly discuss some relevant literary works.

Representations of the problem considered in \cite{tran2019costa, bi2019joint} model the decision making of which services to cache and which tasks to offload as a mixed-integer non-linear optimization problem. In these cases, the problem is NP-hard. Similarly, \cite{chen2017collaborative} models the problem as a graph colouring problem and solves it using parallel Gibbs sampling. While these works try to solve a one-shot cost-minimization problem, in this work we consider the dynamic nature of decision-making based on the input request sequence.

Another model considered in \cite{yan2021pricing} for service hosting focused on the joint optimization of service hosting decision and pricing is a two-stage interactive game between a base-station that provides pricing for edge servers and user equipment which decides whether to offload the task. In another game-theoretic setup, \cite{jiang2020economic, zeng2020novel} delve into the economic aspects of edge caching involving interactions amongst different stakeholders. Some heuristic algorithms have been employed in the works \cite{ascigil2021resource, choi2019latency}. Their approach for the problem is through resource constraint in the latency from the view-point of the edge-cloud infrastructure and not the application provider.

Stochastic models of the system have been considered in \cite{chen2019budget, wang2019dynamic,miao2020intelligent}. While \cite{wang2019dynamic} assumes that the underlying requests follow a Poisson process, \cite{chen2019budget,miao2020intelligent} do not make any prior assumptions regarding the same. \cite{chen2019budget,miao2020intelligent}, through Contextual Combinatorial Multiarmed Bandits aim to use a learning-based approach to make decisions. \cite{wang2019dynamic} formulates the service migration problem as a Markov decision process (MDP) to design optimal service migration policies. These models do not provide any worst-case guarantees for the algorithm, simply average guarantees. In our work, we aim to provide both performance guarantees which are crucial to sensitive applications with large variations in the arrival patterns.

Closest to our work, \cite{narayana2021renting, narayana2021online, prakash2020partial} consider the setting where a service is always hosted at the cloud and consider the algorithmic task of determining when to host the service at the edge server as a function of the arrival process and various system parameters. The key difference between our work and \cite{narayana2021renting, narayana2021online, prakash2020partial} is that, in \cite{narayana2021renting, narayana2021online, prakash2020partial}, once the service is hosted at the edge, the amount of edge computation resources available for use by the service is either fixed or effectively unlimited. Our model allows us to choose the level of computation resources to rent which is a feature available in popular third-party storage/computation resource providers like AWS and Microsoft Azure. Another critical difference between our model and the \cite{narayana2021renting, narayana2021online, prakash2020partial} is the fact that we consider the setting where a non-zero switch cost is incurred every time we change the level of computation resource rented. Contrary to this, in \cite{narayana2021renting, narayana2021online, prakash2020partial}, switch cost is unidirectional, i.e., a switch cost is incurred only when a service is not hosted at the edge in a time-slot and has to be fetched from the cloud servers to host on the edge server. Due to this, the algorithms proposed in \cite{narayana2021renting, narayana2021online, prakash2020partial} and their performance analyses do not directly extend to our setting. Other works on the service hosting problem include \cite{zhao2018red}. At a high level, our work is related to the rich body of work on caching \cite{mukhopadhyay2021online,bhattacharjee2020fundamental,borst2010distributed, belady1966study}.  

\section{Setting}
\label{sec:setting}

We study a system consisting of a cloud server and a third-party-owned edge server. We focus on the problem of efficiently using edge resources from the perspective of a specific service provider. This service is hosted both at the edge and on the cloud server. Each user query/request is routed either to the edge or the cloud and the answer to the query is computed at that server and communicated back to the user, thus necessitating computation power both at the edge and at the cloud servers. We consider a time-slotted setting where the amount of edge computation power rented by the service provider can be changed over time via short-term contracts. 

\emph{Request arrival process}: We consider adversarial and stochastic arrivals. Under the adversarial setting, we make no structural assumptions on the number of requests arriving over time. For our analytical results for stochastic arrivals, we consider the setting where arrivals are i.i.d. over time. 
\begin{assumption}\label{assum_stochastic}(i.i.d. stochastic arrivals)
Let $X_t$ be the number of requests arriving in time-slot $t$. Then, for all $t$,
$\mathbb{P}(X_t = x) = p_x \text{ for } x = 0, 1, 2, \cdots.$
\end{assumption}

In the Gilbert-Elliot-like Model, we make the following assumption:

\begin{assumption}
\label{assum_GE}(Gilbert-Elliot (GE) Model)
Using \cite{GBmodel} as a basis, we consider an arrival process governed by a two-state Markov chain, $A_H$ and $A_L$. Transitions from state $A_H \rightarrow A_L$ and from state $A_L \rightarrow A_H$ occur with probabilities $p_{HL}$  and $p_{LH}$. The state transition diagram has been described in \ref{fig:GEMarkov}. We refer to the two states as the high state and the low state. Under the GE model, if the Markov chain is in the high state, the requests arrive as Poisson($\lambda_H$), and they are Poisson($\lambda_L$) otherwise. 

\begin{figure}
    \centering
    \vspace*{-15mm}
    \begin{tikzpicture}[node distance={30mm}, main/.style = {draw, circle}] 
\node[main] (1) {$A_H$}; 
\node[main] (2) [right of=1] {$A_L$}; 
\draw[->] (1) to [out=45,in=135,looseness=1.5, "$p_{HL}$"] (2);
\draw[->] (2) to [out=225,in=315,looseness=1.5, "$p_{LH}$"] (1);
\draw[->] (1) to [out=135,in=225,looseness=15, "1-$p_{HL}$"] (1);
\draw[->] (2) to [out=45,in=315,looseness=15, "1-$p_{LH}$"] (2);
\end{tikzpicture} 
\vspace*{-10mm}
    \caption{Gilbert-Elliot Model as a Markov Chain}
    \label{fig:GEMarkov}
\end{figure}
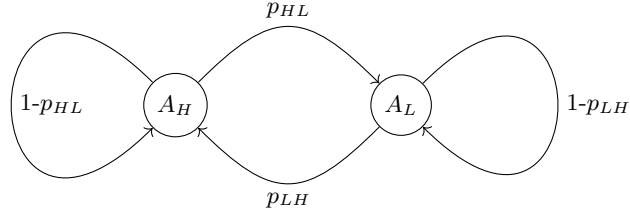
\end{assumption}

\emph{Sequence of events in each time-slot}: We first have request arrivals. These requests are served by the edge server subject to constraints due to limited computation power at the edge. The remaining requests, if any, are forwarded to the cloud server for service. The system then makes a decision on how much edge computation power to rent for the next time-slot.



The algorithmic challenge is to determine how much edge computation power to rent over time. Let $\mathcal{P}$ be a candidate policy that determines the amount of computation power rented by the service provider over time. \\

\subsection{Cost Model and Constraints}
We build on the assumptions in \cite{narayana2021renting, narayana2021online, prakash2020partial}. Under policy $\mathcal{P}$, the service provider incurs three types of costs.
\begin{enumerate}
	\item[--] \emph{Rent cost $(C_{R,t}^\mathcal{P})$}: The service provider can choose one of two possible levels of edge computation power to rent in each time-slot, referred to as high ($H$) and low ($L$). The rent cost incurred per time-slot for levels $H$ and $L$ are denoted by $c_H$ and $c_L(<c_H)$ respectively. 

\item[--]  \emph{Service cost $(C_{S,t}^\mathcal{P})$}: This is the cost incurred due to the latency in service user requests. Given the proximity of the edge servers and the users, no service cost is incurred for requests served at the edge. A cost of one unit is levied on each request forwarded to the cloud server. The highest number of requests that can be served at the edge at edge computation power levels $H$ and $L$ are denoted by $\kappa_H$ and $\kappa_L(<\kappa_H)$ respectively.

\item[--]  \emph{Switch cost $(C_{W,t}^\mathcal{P})$}: Switching from edge computation power level $H$ to $L$ and $L$ to $H$ results in a switch cost of $W_{HL}$ and $W_{LH}$ units respectively.

\end{enumerate}
The number of requests that can be served by the edge server in a time-slot is limited to $\kappa_H$ for state $S_H$ and $\kappa_L$ for state $S_L$ in $\mathbb{Z}^+$, where $\mathbb{Z}^+$ is the set of all positive integers. 
Let $r_t \in \{H, L\}$ denote the edge computation power rented during time-slot $t$ and $X_t$ denote the number of request arrivals in time-slot $t$. It follows that

\begin{align}
\label{equation:total_cost}
C_t^\mathcal{P} &=C_{R,t}^\mathcal{P}+C_{S,t}^\mathcal{P}+C_{W,t}^\mathcal{P},  
\end{align}
\begin{align}
\text{where, }
C_{R,t}^\mathcal{P} &=
\begin{cases*}
c_H & \text{ if $r_{t} = H$ } \nonumber \\
c_L & \text{ if $r_{t} = L$ }
\end{cases*} \\
C_{S,t}^\mathcal{P} & =
\begin{cases*}
X_t-\min\{X_t,\kappa_H\} & \text{ if $r_t = H$ } \nonumber \\
X_t-\min\{X_t,\kappa_L\} & \text{ if $r_t = L$ } 
\end{cases*} \nonumber\\
C_{W,t}^\mathcal{P}&=
\begin{cases*}
W_{HL} & \text{ if $r_{t-1} = H$ and $r_{t} = L$ }  \nonumber\\
W_{LH} & \text{ if $r_{t-1} = L$ and $r_{t} = H$ }  \nonumber\\
0 &\text{ otherwise.} 
\end{cases*}
\end{align}

\begin{remark}
 We limit our discussion to the case where $\kappa_H - \kappa_L > c_H - c_L $. If $\kappa_H - \kappa_L \leq c_H - c_L $, the optimal policy is to always use computation level $L$.
\end{remark}

\subsection{Performance metrics}
We use the following metrics for adversarial and stochastic request arrivals. 

For adversarial arrivals, we compare the performance of a policy $\mathcal{P}$ with the performance of the optimal offline policy (OPT-OFF) which knows the entire arrival sequence a priori. The performance of policy $\mathcal{P}$ is characterized by it competitive ratio $\rho^{\mathcal{P}}$ defined as 
\begin{equation}
\label{eq:competitiveRatio}
\rho^{\mathcal{P}}=\sup_{a\in \mathcal{A}} \frac{C^{\mathcal{P}}(a)}{C^{\text{OPT-OFF}}(a)},
\end{equation}
where $\mathcal{A}$ is the set of all possible finite request arrival sequences, and $C^{\mathcal{P}}(a)$ and $C^{\text{OPT-OFF}}(a)$ are the total costs of service for the request arrival sequence $a$ under the policy $\mathcal{P}$ and the optimal offline policy respectively. 

For i.i.d. stochastic arrivals, we compare the performances of a policy $\mathcal{P}$ with the optimal online policy (OPT-ON) which might know the statistics of the arrival process, but does not know the sample path. The performance metric $\sigma^{\mathcal{P}}_T$ is defined as the ratio of the expected cost incurred by policy $\mathcal{P}$ in $T$ time-slots to that of the optimal online policy in the same time interval. Formally,
\begin{equation}
\label{eq:efficiencyRatio}
\sigma^\mathcal{P}(T) =\frac{\mathbb{E}\bigg[\displaystyle\sum_{t=1}^T C_t^\mathcal{P}\bigg]}{\mathbb{E}\bigg[\displaystyle\sum_{t=1}^T C_t^{\text{OPT-ON}}\bigg]},
\end{equation}
where $C_t^\mathcal{P}$ is as defined in \eqref{equation:total_cost}.


\emph{Goal}: The goal is to design online policies with provable performance guarantees for both adversarial and stochastic arrivals.

\section{Policies}
\label{sec:policies}
In our analysis, we focus the discussion towards \emph{online} policies. At each time-slot, a singular decision must be made determining whether to switch states.

\subsection{Better Late than Never (BLTN)} 
The BLTN policy is inspired by the RetroRenting policy proposed in \cite{narayana2021renting}. BLTN is a deterministic policy that uses recent arrival patterns to evaluate decisions by checking if it made the correct choice in hindsight. 	
Let $t_{\text{switch}}<t$ be the most recent time when the state was changed from $H \rightarrow L$ or $L \rightarrow H$
under BLTN. The policy searches for a time-slot $\tau$ such that $t_{\text{switch}} < \tau < t$, and the total cost incurred is lower if the state is switched in time-slot $\tau-1$ 
and switched back in time-slot $t$ than the cost incurred if the state is not changed during time-slots $\tau-1$ to $t$. If there exists such a time $\tau$, BLTN switches the state in time-slot $t$.

Consider a scenario where the state in time slot $t$ is $S_H$. Let $t_{\text{switch}}<t$ be the time when the server had last changed state to $S_L$ under BLTN. Let $H_i$ and $L_i$ denote cost incurred in time slot $i$ where $H_i$ and $L_i$ are evaluated for $r_i = r_{i-1} = H$ and $L$ respectively. Analytically, the decision to switch to state $S_L$ is made if the algorithm can find a time $\tau$ such that 
	\begin{equation*}
	    {W_{LH} + W_{HL} + \sum_{t_{\text{switch}} \leq i < \tau} H_i + \sum_{\tau \leq j \leq t} L_j < \sum_{t_{\text{switch}} \leq i \leq t} H_i}
	\end{equation*}
	
	\begin{equation}
	\label{algo:line8}
	\text{which simplifies to }
	    {W_{LH} + W_{HL} < \sum_{ i = \tau }^{t} (H_i - L_i).}
	\end{equation}
	
A similar analytical condition can be made for the decision to switch from $S_L$ to state $S_H$. The decision is made if the algorithm can find a time-slot $\tau$ such that 
\begin{equation}
	\label{algo:line16}
	    {W_{LH} + W_{HL} < \sum_{ i = \tau }^{t} (L_i - H_i).}
	\end{equation}

A naive implementation of the algorithm has been constructed in the Appendix. 

While a naive implementation of the BLTN policy can have $\mathcal{O}(T)$ space and time complexity, using techniques proposed in \cite{narayana2021renting}, the time and computational complexity can be reduced to $\mathcal{O}(1)$ as shown through Algorithm \ref{algo:BLTNE}.

\begin{algorithm}
	\caption{Better Late than Never (BLTN)}\label{algo:BLTNE}
	\SetAlgoLined
	Input: Sum of switch costs $W$ units, maximum number of our service requests served by edge server ($\kappa_H$ and $\kappa_L$),
	rent cost: $c_H$ and $c_L$, number of requests: $x_t$, $\underbar{\text{$x$}}_t^H$ = $\min \{x_t, \kappa_H\}$, $\underbar{\text{$x$}}_t^L$ = $\min \{x_t, \kappa_L\}$, $t > 0$\\
	Output:  Service hosting strategy $r_{t+1} \in \{H, L\}$, $t > 0$\\
	Initialize:  Service hosting variable $r_1 = 0$, $\Delta(0)=0$\\
	\For {\textbf{each} time-slot $t$}{
	    $\Delta(t-1) = \Delta(t)$ \\
	    \uIf{$r_t = H$}{
	        $\Delta(t)=\max\big\{0, \Delta(t-1)+\underbar{\text{$x$}}_t^L - \underbar{\text{$x$}}_t^H + c_H - c_L \big\}$\\
	        
    		\uIf{$\Delta(t) > W$}{
    		    $t_{\text{switch}} = t$ \\
    		    $\Delta(t) = 0$ \\
    			$\text{return } r_{t+1}=L$\\
    		}						
    		\Else
    		{
    			$\text{return } r_{t+1}=H$
    		}		
	    }
	    \ElseIf{$r_t = L$}{
	        $\Delta(t)=\max\big\{0, \Delta(t-1)+\underbar{\text{$x$}}_t^H - \underbar{\text{$x$}}_t^L + c_L - c_H\big\}$\\
	        
    		\uIf{$\Delta(t) > W$}{
    		    $t_{\text{switch}} = t$ \\
    		    $\Delta(t) = 0$ \\
    			$\text{return } r_{t+1}=H$\\
    		}						
    		\Else
    		{
    			$\text{return } r_{t+1}=L$
    		}		
	    }
	 }
\end{algorithm}

\subsection{Follow The Perturbed Leader (FTPL)} 
FTPL \cite{mukhopadhyay2021online,bhattacharjee2020fundamental} is a randomized policy. In time-slot $t$, it compares suitably perturbed versions of the cost incurred from time 1 to $t$ under two static decisions, i.e., state $L$ from time 1 to $t$ and state $H$ from time 1 to $t$. The state of the system is then set to the one which has the lower perturbed cost.

\begin{algorithm}
	\caption{Follow the Perturbed Leader (FTPL)}\label{algo:FTPL}
	\SetAlgoLined
	Input: Switch costs $W_{HL}$ and $W_{LH}$ units, maximum number of our service requests served by edge server ($\kappa_H$ and $\kappa_L$),
	rent cost: $c_H$ and $c_L$, number of requests: $x_t$, $\underbar{\text{$x$}}_t^H$ = $\min \{x_t, \kappa_H\}$, $\underbar{\text{$x$}}_t^L$ = $\min \{x_t, \kappa_L\}$, $t > 0$, Gaussian distribution with mean $\mu$ and variance $\sigma$: $\mathcal{N}(\mu, \sigma)$\\
	Output:  Service hosting strategy $r_{t+1} \in \{H, L\}$, $t > 0$\\
	Initialize:  Service hosting variable $r_1 = 0$, $\Delta(0)=0$\\
	\For {\textbf{each} time-slot $t$}{
	    $\Delta(t-1) = \Delta(t)$ \\
	    $\Delta(t)= \Delta(t-1)+\underbar{\text{$x$}}_t^L - \underbar{\text{$x$}}_t^H + c_H - c_L$
	    
    		\uIf{$\Delta(t) + \gamma \mathcal{N}(0, \sqrt{t}) > 0$}{
    			$\text{return } r_{t+1}=L$\\
    		}						
    		\Else
    		{
    			$\text{return } r_{t+1}=H$
    		}		
	    
	 }
\end{algorithm}

The variation of the perturbation $\mathcal{N}(0, \sqrt{t})$ increases as $\sqrt{t}$, while the total difference in cost scaled linearly with time. This implies that the FTPL policy, over time, chooses to remain static in the state with the least cost.

\subsection{An illustration}

We consider a sample sequence of arrivals. Let $r_t$ be the state sequence with time index $t$. We consider the case where $W_{LH} = W_{HL} = 275, c_H = 600, c_L = 400, \kappa_H = 700, \kappa_L = 300$ and a request sequence

\vspace{2mm}
\small{
\noindent
\begin{tikzpicture}
 	\foreach \x in {1,2,...,11}{
 	 \draw[gray] (0.8*\x,-1.2mm) -- (0.8*\x,1.2mm);
 	  \draw[] (0.8*\x,0) -- node[below=0.2mm,pos=0.5] {\x} (0.8*\x+0.8,0);
 		\draw[gray] (0.8*\x+0.8,-1.2mm) -- (0.8*\x+0.8,1.2mm);
 		}
 		\draw[] (1,3mm) -- node[left=2mm,pos=0.25] {Number of requests:} (1,2mm);
 		\draw[] (1,-3mm) -- node[left=2mm,pos=0.25] {Time-slot index:} (1,-2mm);
 		\draw[] (1,0) -- node[above=0.2mm,pos=0.25] {900} (1.8,0);
 		\draw[] (1.8,0) -- node[above=0.2mm,pos=0.25] {900} (2.6,0);
 		\draw[] (2.6,0) -- node[above=0.2mm,pos=0.25] {900} (3.4,0);
 		\draw[] (3.4,0) -- node[above=0.2mm,pos=0.25] {900} (4.2,0);
 		\draw[] (4.2,0) -- node[above=0.2mm,pos=0.25] {900} (5,0);
 		\draw[] (5,0) -- node[above=0.2mm,pos=0.25] {900} (5.8,0);
 		\draw[] (5.8,0) -- node[above=0.2mm,pos=0.25] {200} (6.6,0);
 		\draw[] (6.6,0) -- node[above=0.2mm,pos=0.25] {200} (7.4,0);
 		\draw[] (7.4,0) -- node[above=0.2mm,pos=0.25] {200} (8.2,0);
 		\draw[] (8.2,0) -- node[above=0.2mm,pos=0.25] {200} (9,0);
 		\draw[] (9,0) -- node[above=0.2mm,pos=0.25] {200} (9.8,0);
  	\end{tikzpicture}
}
\vspace{2mm}

\normalsize
Initially, we consider the edge server to be in state $S_L$. We observe that the optimal state to serve 900 incoming requests is $S_H$. While hosting under the BLTN policy, the first switch to state $S_H$ occurs in the time-slot 3, thus $r_4 = H$ and $r_{1, 2, 3} = L$. The cost incurred in the case where the state is $S_L$ till $t=3$ is $\textstyle\sum_{l=1}^3 (x_l - \kappa_L)^+ + (3 - 1 + 1) \times c_L = 3000$, while the cost incurred in state $S_H$ is $\textstyle\sum_{l=1}^3 (x_l - \kappa_H)^+ + (3 - 1 + 1) \times c_H = 2400$. The difference in the cost equates $600 > W = 550$, and that is the first time the condition \ref{algo:line16} is satisfied. $t_{switch}$ is updated to 3, and $r_4 = H$.

From time-slot 4 onward, BLTN hosts in state $S_H$ up to time-slot 8. We set $\tau = 6 > t_{switch}$ and evaluate the condition \ref{algo:line8}. Setting $t = 8, \tau = 6$, we have $\textstyle\sum_{l=6}^8 ((x_l - \kappa_H)^+ - (x_l - \kappa_L)^+) + (8 - 6 + 1) \times (c_H-c_L) = 600 \geq W = 550$. This is the first time-slot since 3 that the condition is satisfied, thus $r_8 = H, r_{9,...} = L$ till the next time BLTN decides to switch.

\section{Main Results and Discussion}
\label{sec:mainResults}
Our first theorem characterizes the performance of BLTN for adversarial arrivals by giving a worst-case guarantee on the performance of the BLTN policy against the optimal offline policy. We also characterize a lower bound performance of any deterministic online policy.
\begin{theorem}
\label{thm:BLTN_adv_online}
Let $\Delta \kappa = \kappa_H - \kappa_L$,  $\Delta c = c_H - c_L$, $ W = W_{LH} + W_{HL}$.
If $\Delta \kappa > \Delta c$ then,  
\begin{enumerate}
   \item[(a)] $\rho^{\text{BLTN}}\leq \left(1 + \frac{2W + \Delta \kappa}{W \left( 1 + \frac{c_H}{\Delta \kappa - \Delta c}+ \frac{c_L}{\Delta c } \right)}\right)$,
   \item[(b)] $\rho^{\mathcal{P}} \geq
\text{min } \left\{ \frac{\Delta \kappa + c_L}{c_H}, \frac{c_H}{c_L}, \frac{\Delta \kappa + c_L + c_H + W}{c_H + c_L + W} \right\}$, for any deterministic policy $\mathcal{P}$.
\end{enumerate}
\end{theorem}

Theorem \ref{thm:BLTN_adv_online} provides a worst-case guarantee on the performance of the BLTN policy against the optimal offline policy. Unlike the BLTN policy, the optimal offline policy has complete information of the entire arrival sequence beforehand. 
We note that the competitive ratio of BLTN improves as the sum of switch costs ($W_{LH} + W_{HL}$) increases.
Also, the competitive ratio of BLTN increases linearly with the difference of the caps on requests served at the edge, $\Delta \kappa$. Supplementing it, we have Theorem \ref{thm:BLTN_adv_online} (b) which shows that the competitive ratio of \emph{any} deterministic online policy increases linearly with $\Delta \kappa$. 
While Theorem \ref{thm:BLTN_adv_online} (a) provides a worst-case guarantee for the BLTN policy, it must be noted that through subsequent simulations, the performance of BLTN is substantially closer to the optimal offline policy.


Next we summarize the performance of the BLTN  policy for i.i.d. stochastic arrivals (Assumption~\ref{assum_stochastic}, Section \ref{sec:setting}).

This lemma gives a bound on the expected difference between the costs incurred in a time-slot by the BLTN policy and the optimal online policy. We use the functions $f(\cdot), g(\cdot)$  which are defined in the Appendix. 
The functions $f(\cdot), g(\cdot)$ are formulated using Hoeffding’s inequality to bound the probability of certain events. We use the functions $f(\cdot), g(\cdot)$ for the sake of compactness.

\small
\begin{lemma}
	\label{lemma:difference_BLTNstochastic}
	Let $\Delta_t^{\mathcal{P}} = \mathbb{E}[C_t^{\mathcal{P}} - C_t^{\text{OPT-ON}}]$, $\mu_H =  \mathbb{E}[\underbar{$X$}_{t,H}]$, $\mu_L =  \mathbb{E}[\underbar{$X$}_{t,L}]$, $\Delta \mu = \mu_H - \mu_L$, $\Delta \kappa = \kappa_H - \kappa_L$, and  $\Delta c = c_H - c_L$.
	
 \begin{align*}
	f(\Delta \kappa,\lambda, W, \Delta \mu, \Delta c)
	= &(W + \Delta \mu - \Delta c) \times \Bigg( 2\bigg\lceil\frac{\lambda W}{\Delta \mu- \Delta c}\bigg\rceil\frac{\exp(-2\frac{(\Delta \mu- \Delta c)^2\frac{W}{\Delta c}}{(\Delta \kappa)^2})} {1-\exp(-2\frac{(\Delta \mu- \Delta c)^2}{ (\Delta \kappa)^2})} \\ &+ \exp(-2\frac{(\lambda-1)^2 W(\Delta \mu- \Delta c)}{\lambda  (\Delta \kappa)^2})\Bigg), \text{ and} \\
	g(\Delta \kappa,\lambda, W, \Delta \mu, \Delta c) 
	= &(\Delta c - \Delta \mu + W) \times \Bigg( \exp(-2\frac{(\lambda-1)^2(\Delta c - \Delta \mu)W}{\lambda  (\Delta \kappa)^2}) \\ &+ 2\bigg\lceil\frac{\lambda W}{\Delta c - \Delta \mu}\bigg\rceil \frac{\exp(-2\frac{(\Delta c - \Delta \mu)^2 \frac{W}{\Delta \kappa- \Delta c}}{ (\Delta \kappa)^2})} {1-\exp(-2\frac{(\Delta c - \Delta \mu)^2}{ (\Delta \kappa)^2})} \Bigg).
	\end{align*}
	Then, under Assumption \ref{assum_stochastic},
	\begin{itemize}
		\item[--] Case $\Delta \mu > \Delta c$:
		\begin{align*}
			\Delta_t^{\text{BLTN}}(\lambda) &\leq \begin{cases}
			W+ \Delta \mu - \Delta c, & t \leq \big\lceil\frac{\lambda W}{\Delta \mu- \Delta c}\big\rceil\\
            f(\Delta \kappa,\lambda, W, \Delta \mu, \Delta c), & t > \big\lceil\frac{\lambda W}{\Delta \mu- \Delta c}\big\rceil
		 \end{cases}.
		\end{align*}
		\item[--] Case $\Delta \mu < \Delta c$:
\begin{align*}
\Delta_t^{\text{BLTN}}(\lambda) &\leq \begin{cases}
			W+ \Delta c - \Delta \mu, & t \leq \big\lceil\frac{\lambda W}{\Delta c- \Delta \mu}\big\rceil\\
            g(\Delta \kappa, \lambda, W, \Delta \mu, \Delta c), & t > \big\lceil\frac{\lambda W}{\Delta c - \Delta \mu}\big\rceil
		 \end{cases}.
\end{align*}
	\end{itemize}
\end{lemma}

\normalsize

We can conclude that for large enough $t$, the difference between the cost incurred by BLTN and the optimal online policy in time-slot $t$, decays exponentially with $W$ and $|\Delta \mu - \Delta c|$.

Theorem \ref{thm:BLTN_stochastic_theorem} gives an upper bound on the ratio of the expected cost incurred under BLTN and the optimal online policy in a setting where request arrivals are stochastic. In the statement of this theorem, we used the functions $f$ and $g$  which are defined in Lemma \ref{lemma:difference_BLTNstochastic}.

\small
\begin{theorem}\label{thm:BLTN_stochastic_theorem}
		Let 
	$\nu = \mathbb{E}[X_t],$ $\mu_H =  \mathbb{E}[\underbar{$X$}_{t,H}]$, and $\mu_L =  \mathbb{E}[\underbar{$X$}_{t,L}]$. 
	Let the rent cost per time-slot be $c_H$ or $c_L$ depending on the states $S_H$ or $S_L$ respectively. Define $\Delta \mu = \mu_H - \mu_L$, $\Delta \kappa = \kappa_H - \kappa_L$,  $\Delta c = c_H - c_L$.
	Recall the definition of $\sigma^{\mathcal{P}}_T$ given in \eqref{eq:efficiencyRatio}.

		\begin{itemize}
		\item[--] Case $\Delta \mu > \Delta c$: For the function $f$ defined in Lemma \ref{lemma:difference_BLTNstochastic},
		
		\begin{align*}
		\hspace*{-7mm}
		\sigma^{\text{BLTN}}(T) \leq \min_{\lambda>1} \Bigg(1 + \dfrac{\big\lceil \frac{\lambda W}{\Delta \mu- \Delta c}\big\rceil (W+ \Delta \mu - \Delta c)}{T(\nu - \mu_H + c_H)} 
		+ \dfrac{\left(T - \big\lceil \frac{\lambda W}{\Delta \mu- \Delta c}\big\rceil \right)f(\Delta \kappa,\lambda,W, \Delta \mu, \Delta c)}{T(\nu - \mu_H + c_H)}\Bigg),
			\end{align*}

		\item[--] Case $\Delta \mu < \Delta c$: For the function $g$ defined in Lemma \ref{lemma:difference_BLTNstochastic},
		\begin{align*}
		\hspace*{-7mm}
		\sigma^{\text{BLTN}}(T) \leq \min_{\lambda>1} \Bigg(1 + \dfrac{\big\lceil \frac{\lambda W}{\Delta c- \Delta \mu}\big\rceil (W+ \Delta c - \Delta \mu)}{T(\nu - \mu_L + c_L)} 
		+ \dfrac{\left(T - \big\lceil \frac{\lambda W}{\Delta c- \Delta \mu}\big\rceil \right)g(\Delta \kappa,\lambda,W, \Delta \mu, \Delta c)}{T(\nu - \mu_L + c_L)}\Bigg).
		\end{align*}
		
	\end{itemize}
\end{theorem}
\normalsize

We observe that the bounds in Lemma \ref{lemma:difference_BLTNstochastic} worsen with an increase in $\Delta \kappa$. It must be noted that this is a bound obtained using Hoeffding's inequality which does not assume any specific i.i.d. process (Chernoff bound presents a stronger inequality here). For generic cases, the performance of BLTN does not worsen with an increase in $\Delta \kappa$ as shown via simulations in the next section.

For large values of $T$, the bound on the ratio of total expected costs reduces exponentially with the sum of switch costs $W$. The performance guarantees obtained for BLTN in this section show that BLTN performs well in both the general and the i.i.d. stochastic settings without making any assumptions on the request arrival process.

\section{Simulations}
\label{sec:simulations}
Since our analytical results only provide bounds on the BLTN policy, we now compare the performance of BLTN, FTPL, and the optimal online policy which uses the knowledge of the statistics of the arrival process to make decisions via simulations. Recall that both BLTN and FTPL do not know the statistics of the arrival process. 

\label{sec:simulations}
Unless stated otherwise, the parameter values in the simulations are as follows: $c_L = 300$;  $\Delta c = 300$; $k_L = 400$; $\Delta \kappa = 400$; $W_{HL} = W_{LH} = 300$. For the GE model, the transition probability from either state of the two-state Markov chain is 0.01. In Algorithm \ref{algo:FTPL}, we set $\gamma = 500$, through empirical observations of overall performance.  All the simulations have been averaged over the same set of 50 random seeds over 10,000 time-slots. The captions highlight the arrival sequence model - Assumption \ref{assum_stochastic} (i.i.d. Poisson) or Assumption \ref{assum_GE} (GE).




\vspace*{-10mm}
\begin{figure}[H]
\centering
  \subfloat[ i.i.d. Poisson]{%
   \includegraphics[width=0.48\linewidth]{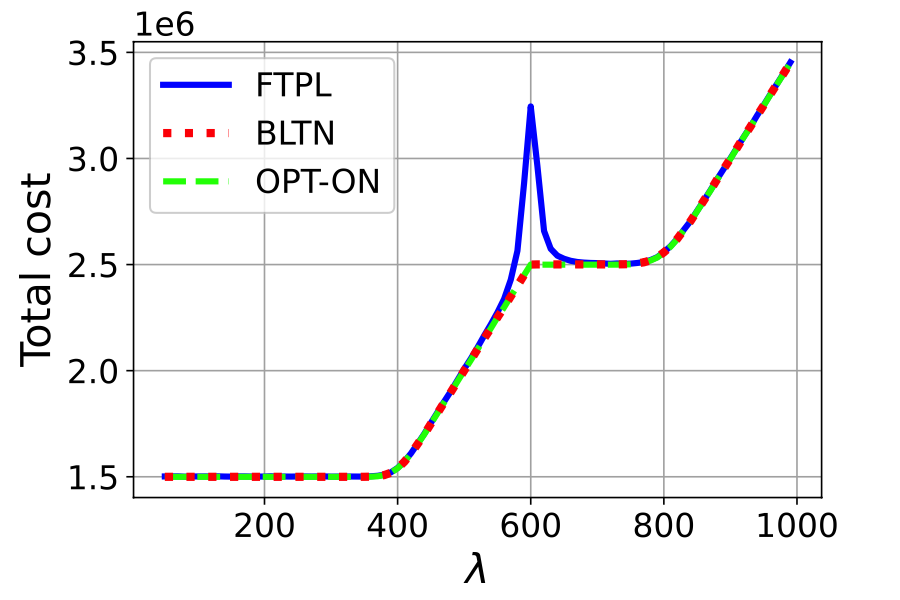}}
    \hspace{2pt}
    \subfloat[GE]{%
   \includegraphics[width=0.48\linewidth]{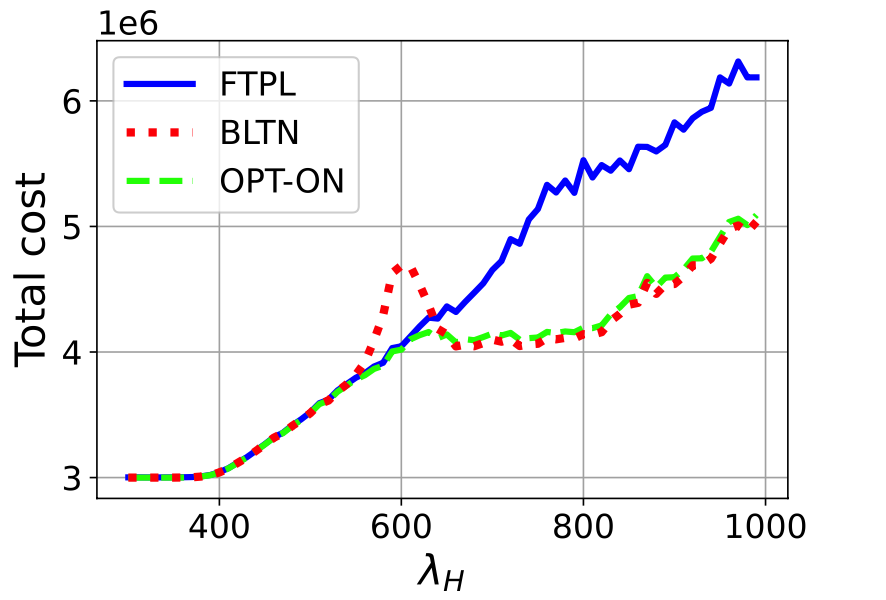}}
  \caption{\small Performance of various policies as a function of the request arrival rate. For the GE model, we fix $\lambda_L = 300$ and vary $\lambda_H$.}
  \label{fig:lambda1} 
\end{figure}

\vspace{-5mm}
In Figure \ref{fig:lambda1}, we see that for i.i.d. Poisson arrivals, the performance of BLTN matches that of FTPL for most of the $\lambda$ values considered except around $\lambda = 600$. At this value of $\lambda$, the expected cost incurred at levels $L$ and $H$ is very close and as a result, the switch cost incurred by FTPL is high, thus leading to poor performance. We note that for the GE model, BLTN outperforms FTPL for a large range of $\lambda_H$. The superior performance of BLTN is a consequence of the fact that unlike FTPL, BLTN puts added emphasis on recent arrival patterns when making decisions. The same trend follows in Figure \ref{fig:k}. 

\vspace*{-10mm}
\begin{figure}[H]
\centering
  \subfloat[ i.i.d. Poisson, $\lambda = 700$]{
    \includegraphics[width=0.48\linewidth]{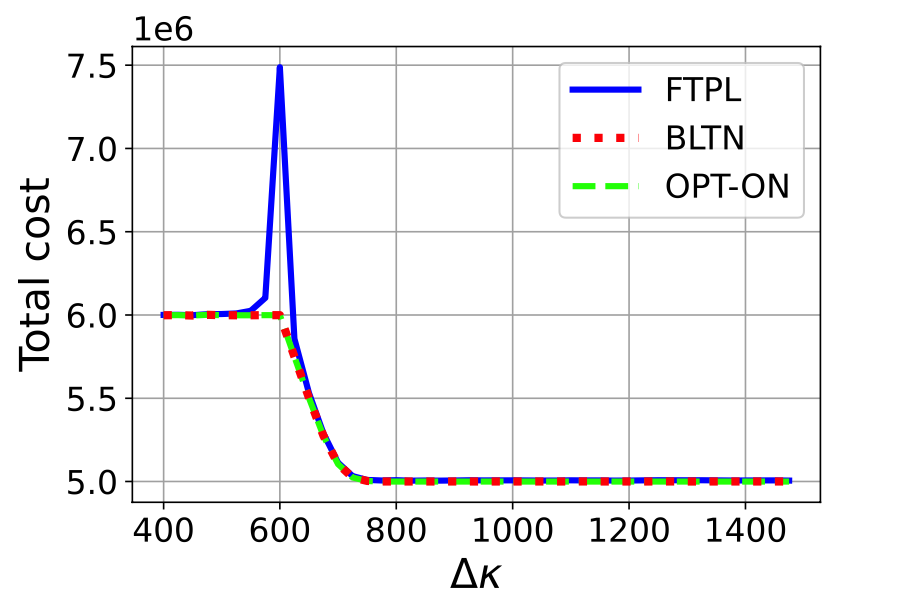}}
    \hspace{2pt}
    \subfloat[\small{$\text{GE}$, $\lambda_H = 800$, $\lambda_L = 300$}]{%
   \includegraphics[width=0.48\linewidth]{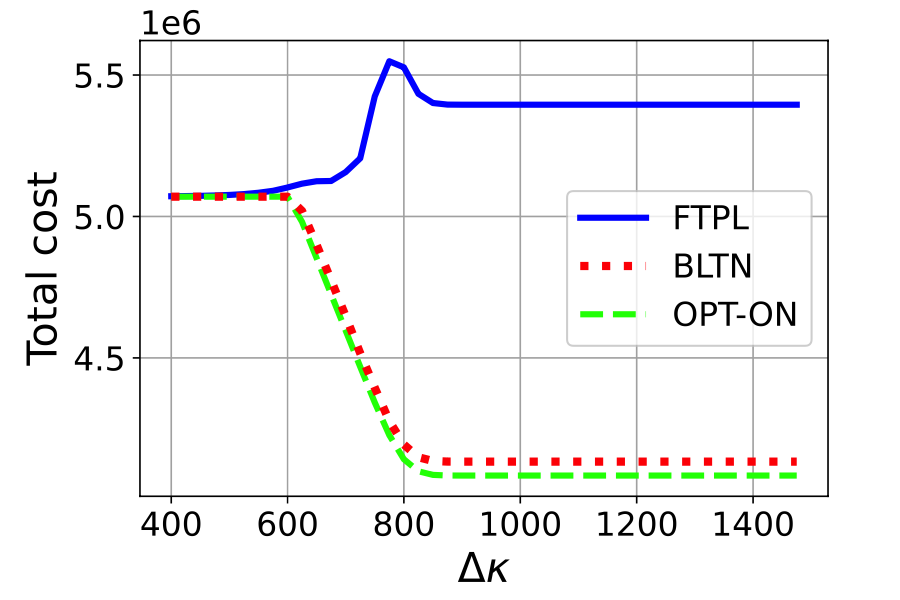}}
   \caption{ \small Performance of various policies as a function of $\Delta \kappa$}
  \label{fig:k} 
\end{figure}

\vspace*{-25mm}
\begin{figure}[H]
\centering
  \subfloat[ $\text{Poisson}(\lambda = 700)$]{%
   \includegraphics[width=0.48\linewidth]{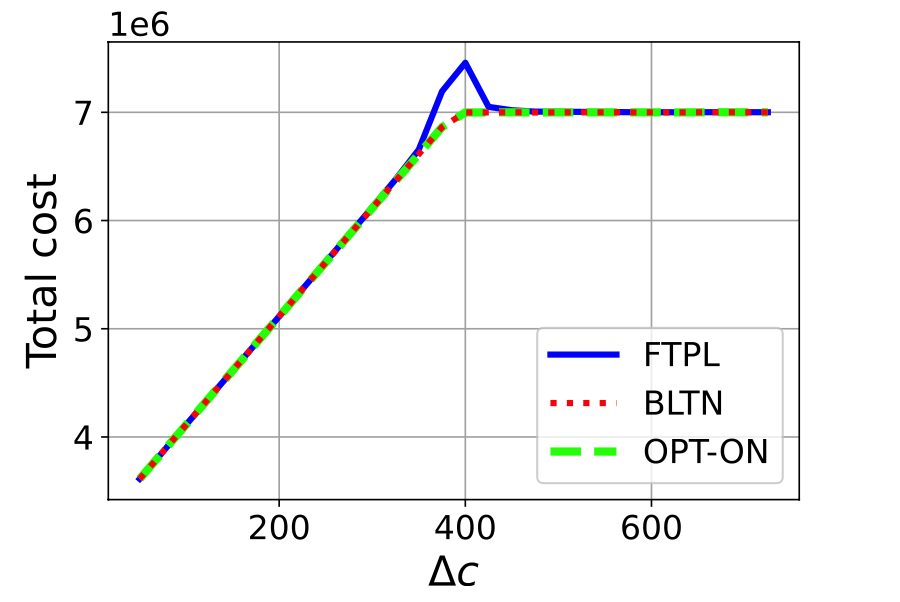}}
    \hspace{2pt}
    \subfloat[$GE(\lambda_H = 800,\lambda_L = 300)$]{%
   \includegraphics[width=0.48\linewidth]{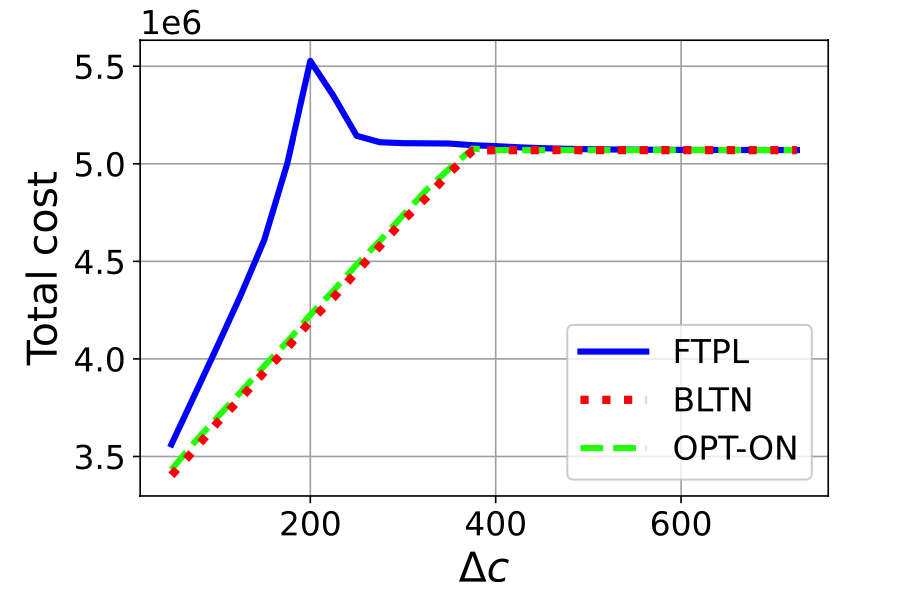}}
  \caption{\small Performance of various policies as a function of difference in rent costs $\Delta c$}
  \label{fig:c} 
\end{figure}

\vspace*{-25mm}
\begin{figure}[H]
\centering
  \subfloat[ $\text{Poisson}(\lambda = 700)$]{%
   \includegraphics[width=0.48\linewidth]{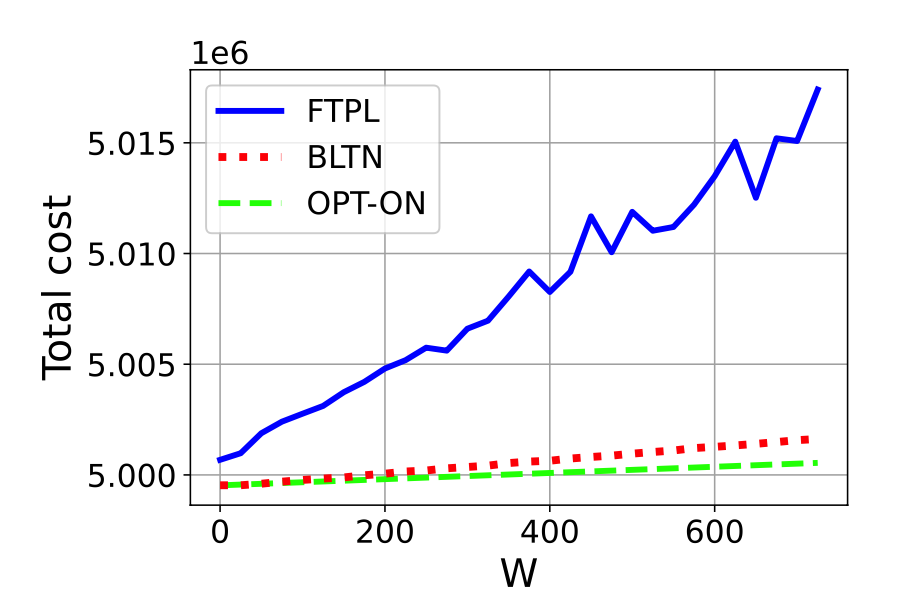}}
    \hspace{2pt}
    \subfloat[$GE(\lambda_H = 800,\lambda_L = 300)$]{%
   \includegraphics[width=0.48\linewidth]{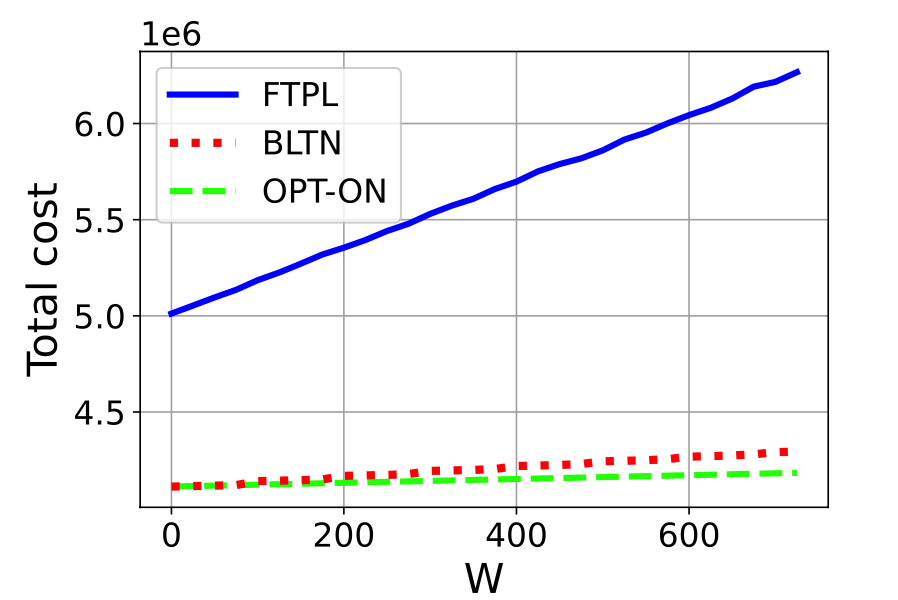}}
  \caption{\small Performance of various policies as a function of switch costs, $W = W_{HL} = W_{LH}$}
  \label{fig:W} 
\end{figure}

\vspace*{-15mm}
\begin{figure}[H]
\centering
  \subfloat[ $\text{Poisson}(\lambda = 700)$]{%
   \includegraphics[width=0.48\linewidth]{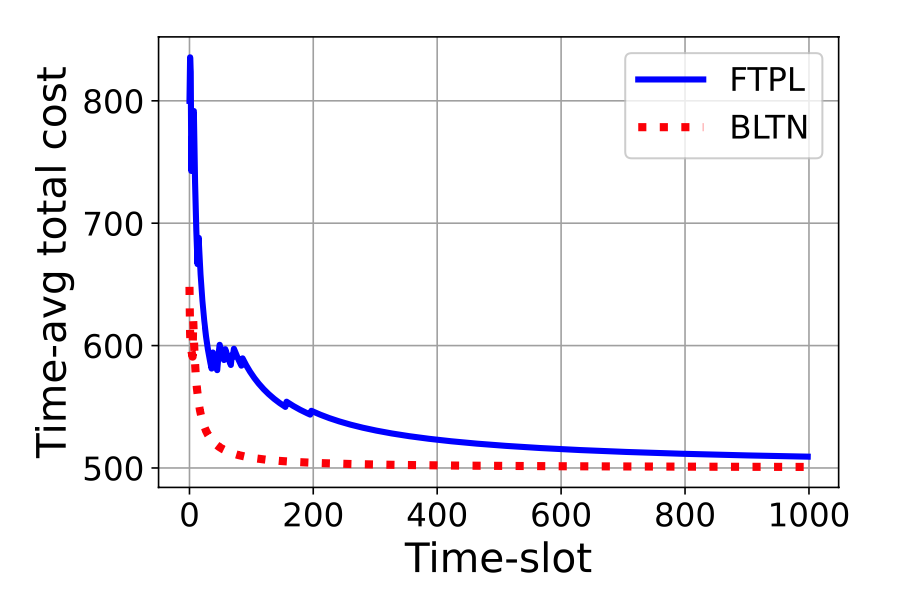}}
    \hspace{2pt}
    \subfloat[$GE(\lambda_H=800,\lambda_L = 300)$]{%
   \includegraphics[width=0.48\linewidth]{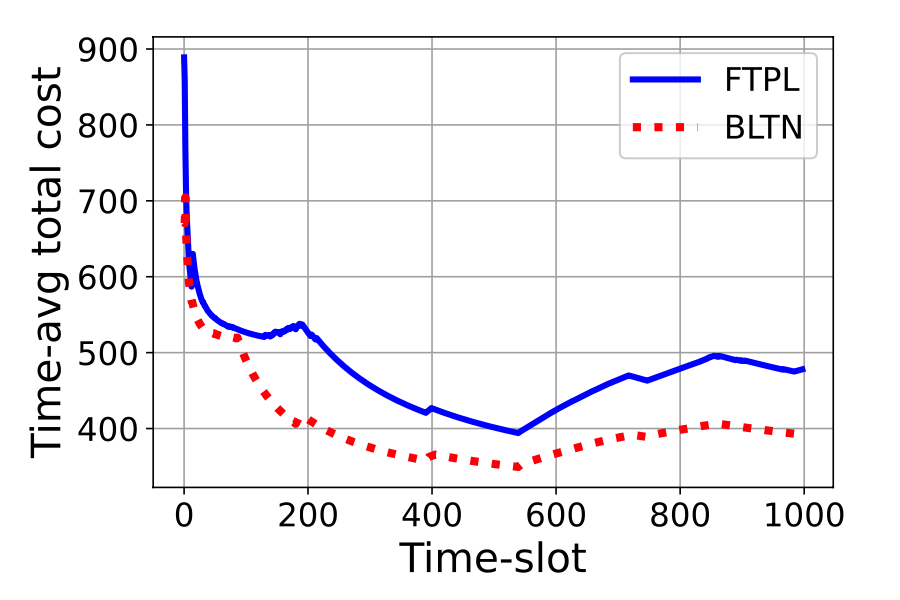}}
  \caption{Time averaged total cost}
  \label{fig:time-1000} 
\end{figure}


Through Theorems \ref{thm:BLTN_adv_online} and \ref{thm:BLTN_stochastic_theorem}, it is suggested that the bounds worsen as $\Delta \kappa$ increases. However, under Assumption \ref{assum_stochastic}, there is significant difference in performance only when $\Delta \kappa = \Delta c$, and under Assumption \ref{assum_GE}, there is significant difference whenever $\Delta \kappa > \Delta c$.

In arrival sequences characterized by Assumption \ref{assum_GE}, usually BLTN performs better owing to its ability to draw conclusions from history. FTPL fails to account for switching costs and so, is sub-optimal.

\section{Conclusion}
\label{sec:conclusion}
We consider the problem of renting edge computing resources for serving customer requests at the edge. We propose an online policy called Better-Late-Than-Never (BLTN) and provide performance guarantees for adversarial and stochastic request arrivals. Further, we compare the performance of BLTN with the widely studied FTPL policy. We conclude that BLTN outperforms FTPL for most settings considered, especially when the statistics of the arrival process are time-varying. The main reason for this is that BLTN makes decisions based on recent request arrival patterns while FTPL uses the entire request arrival history to make decisions.

\bibliographystyle{splncs04}
\bibliography{main}

\newpage
\section*{Appendix}
\label{sec:appendix}



\section{BLTN Naive Algorithm}
\begin{algorithm}
	\caption{Better Late Than Never (BLTN)}\label{algo:BLTN}
	\SetAlgoLined
	
	Input: Sum of switch costs $W$ units, maximum number of our service requests served by edge server in states $S_H$ and $S_L$ as $\kappa_H$ and $\kappa_L$,  rent costs $c_H$ and $c_L$, request arrival sequence: $\{x_l\}_{l=0}^t$, $t > 0$\\
	Output:  Service switching strategy $r_{t+1}$, $t > 0$\\
	Initialize:  Service switching variable $r_1 = L$\\
	\For {\textbf{each} time-slot $t$}{
		$r_{t+1}=r_t$\\
		
		\If{$r_t=H$ }{
			\For{$t_{\text{switch}} < \tau < t$}{
				\If{$\displaystyle\sum_{l=\tau}^t (x_l - \kappa_H)^+ + (t - \tau + 1) \times c_H \geq \displaystyle\sum_{l=\tau}^t(x_l-\kappa_L)^+ + (t - \tau + 1) \times c_L + W,$}{
					$r_{t+1}=L$, $t_{\text{switch}} =  t$\\
					break\\
				}			
			}
		}
		\If{$r_t=L$ }{ 
			\For{$t_{\text{switch}} < \tau < t$}{
				\If{$\displaystyle\sum_{l=\tau}^t (x_l - \kappa_L)^+ + (t - \tau + 1) \times c_L \geq \displaystyle\sum_{l=\tau}^t(x_l-\kappa_H)^+ + (t - \tau + 1) \times c_H + W,$}{
					$r_{t+1}=H$, $t_{\text{switch}} =  t$\\
					break\\
				}			
			}		
		}
		
	}
\end{algorithm}

\section{Proof Outlines}
\label{sec:proofoutlines}
In this section we outline the proofs of the results discussed in Section \ref{sec:mainResults}. The proof details follow. 

\subsection{Proof Outline for Theorem \ref{thm:BLTN_adv_online} (a)}
The time axis is partitioned into `frames'. Frame $i$ for $i \in \mathbb{Z}^+$ begins when OPT-OFF switches the state of the edge-server for the $2i^{\text{th}}$ time. The time interval before the first frame is denoted as Frame 0. 

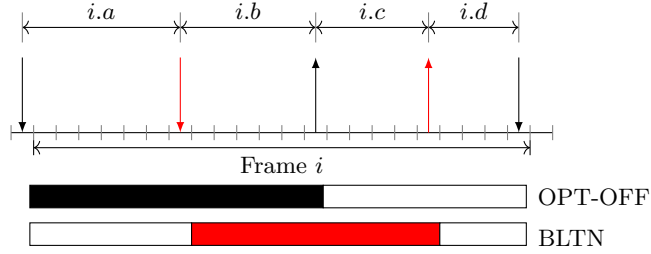
\begin{figure}[ht]
	\centering
	\begin{tikzpicture}
	\foreach \x in {0,0.3,0.6,...,7.2}{
		\draw[] (\x,0) --  (\x+0.3,0);
		\draw[gray] (\x,-1mm) -- (\x,1mm);
	}
	
	\draw[<->,color=black] (0.15,14mm) -- 
	node[above=0mm,pos=0.5]{$i.a$} (2.25,14mm);
	\draw[<->,color=black] (2.25,14mm) -- 
	node[above=0mm,pos=0.5]{$i.b$} (4.05,14mm);
	\draw[<->,color=black] (4.05,14mm) -- 
	node[above=0mm,pos=0.5]{$i.c$} (5.55,14mm);
	\draw[<->,color=black] (5.55,14mm) -- 
	node[above=0mm,pos=0.5]{$i.d$} (6.75,14mm);
	
	\foreach \x in {0.15,2.25,4.05,5.55,6.75}{
		\draw[gray] (\x,12mm) -- (\x,16mm);
	}
	\draw[gray] (7.2,-1mm) -- (7.2,1mm);
	\draw[-latex] (0.15,10mm) -- node[above=5mm]{} (0.15,0mm);
	\draw[-latex] (6.75,10mm) -- node[above=5mm]{} (6.75,0mm);
	\draw[-latex] (4.05,0mm) -- node[above=5mm]{} (4.05,10mm);
	\draw[-latex,color=red] (2.25,10mm) -- node[above=5mm]{} (2.25,0mm);
	\draw[-latex,color=red] (5.55,0mm) -- node[above=5mm]{} (5.55,10mm);
	
	\draw[<->,color=black] (0.3,-2mm) -- node[below=0mm,pos=0.5]{Frame 
		$i$} (6.9,-2mm);
	
	\draw[black] (0.3,-1mm) -- (0.3,-3mm);
	\draw[black] (6.9,-1mm) -- (6.9,-3mm);
	
	\node[] at (7.75,-0.85) {OPT-OFF};
	
	\filldraw[fill=black!] (0.25,-1) rectangle (4.15,-0.7);
	\filldraw[fill=white!40] (4.15,-1) rectangle (6.85,-0.7);
	\node[] at (7.45,-1.4) {BLTN};
	\filldraw[fill=white!40] (0.25,-1.5) rectangle (2.4,-1.2);
	\filldraw[fill=red!] (2.4,-1.5) rectangle (5.7,-1.2);
	\filldraw[fill=white!40] (5.7,-1.5) rectangle (6.85,-1.2);

	\end{tikzpicture}
	\caption{Figure representing 
	switching by OPT-OFF and BLTN in the $i^{\text{th}}$ frame. }
	\vspace*{-5mm}
\label{fig:OPT-OFF_BLTN_frame_proofoutline}
\end{figure}

In figure \ref{fig:OPT-OFF_BLTN_frame_proofoutline}, the downward arrows represent switches to state $S_H$, upward arrows indicate switches to state $S_L$.
	We designate red and black arrows to the BLTN and OPT-OFF policies respectively. 
	The state of the system under policies OPT-OFF and BLTN have been indicated using the two bars below the time axis. The solid red and solid black portions represent the 
	intervals during which BLTN and OPT-OFF host the service in state $S_H$ respectively.

We use the properties of BLTN and OPT-OFF to show that each frame has the following structure (Figure \ref{fig:OPT-OFF_BLTN_frame_proofoutline}):
\begin{enumerate}
	\item[--] BLTN switches state back and forth exactly once. 
	\item[--] BLTN starts in state $S_L$.
    \item[--] The switch to $S_H$ by BLTN in Frame $i$ is before OPT-OFF switches to state $S_L$ in Frame $i$.
    \item[--] The switch to $S_L$ by BLTN in Frame $i$ is after OPT-OFF switches to $S_L$ in Frame $i$.
\end{enumerate}
We note that both BLTN and OPT-OFF switch exactly once back and forth in a frame and therefore, the total switch cost under BLTN and OPT-OFF is identical for both policies. 


From Lemma 5 
, we have that after switching to state $S_H$, OPT-OFF hosts the service for at least $\frac{W}{\Delta \kappa-\Delta c}$ time-slots. Similarly, after switching to $S_L$, OPT-OFF hosts the service for at least $\frac{W}{\Delta c}$ (Lemma 6) time-slots.

Across subframes $i.b$ and $i.d$, the rent and service costs are the same for both the policies. However, the differences between the cumulative service and rent costs incurred by BLTN and OPT-OFF in subframes $i.a$ and $i.c$ are capped by $W + \Delta \kappa - \Delta c$ and $W + \Delta c$ respectively. Also, the total switch cost under both policies in a frame is the same. Using these results, we have that the total cost incurred by BLTN and OPT-OFF in a frame differs by at most $2W + \Delta \kappa$.

In case of the last frame, if OPT-OFF switches the service, the analysis is the same as that of the previous frame. Otherwise, we bound the ratio of the cost incurred by BLTN and cost incurred by OPT-OFF in the frame.

Combining the results for individual frames obtained earlier, the final result follows.

\subsection{Proof Outline for Theorem \ref{thm:BLTN_adv_online} (b)}

All deterministic online policies can be partitioned into two subsets.
A policy that hosts the service in $S_H$ during the first time-slot is in the first subset. All other polices are in the second subset.

In either of the subset, for each policy we construct an arrival sequence and calculate the ratio of the cost of the deterministic online policy and an alternative policy. Following from the definition, this quantity is a lower-bound on the competitive ratio of the deterministic policy.

\subsection{Proof Outline for Theorem \ref{thm:BLTN_stochastic_theorem}}

Through Lemma 15 
, we identify a lower bound on the cost per time-slot incurred by any online policy.

First, we consider the case where $\Delta \mu > \Delta c$. Through Hoeffding's inequality \cite{hoeffding1994probability}, we upper bound the probability of the state being $S_L$ during time-slot $t$ under BLTN. Conditioned on being in state $S_H$ and not switching to $S_L$ in time-slot $t$, the expected cumulative cost incurred by BLTN  is at the most $c_H+\mathbb{E}[X_t - \min\{X_t, \kappa_H\}]$ and is upper bounded by $W + c_L + \mathbb{E}[X_t - \min\{X_t, \kappa_L\}]$ otherwise. The result follows.

Next, we consider the case where $\Delta \mu < \Delta c$. We upper bound the probability of the service being hosted in state $S_H$ in time-slot $t$ under BLTN through Hoeffding's inequality  \cite{hoeffding1994probability}. 
Similarly, the result then follows by the fact that the expected total cost incurred by BLTN is no more than $c_L + \mathbb{E}[X_t - \min\{X_t, \kappa_L\}]$ and at the most $W + c_H + \mathbb{E}[X_t - \min\{X_t, \kappa_H\}]$ otherwise, conditioned on being in state $S_L$ and not switched to $S_H$ in time-slot $t$.

\label{sec:proofs}
\section{Proofs}
\label{sec:proofthm1a}

\subsection{Proof of Theorem \ref{thm:BLTN_adv_online}(a)}

The notation used in this subsection is given in Table \ref{table:proofs}. 
\begin{table}
	\centering
	{\renewcommand{\arraystretch}{1.4}%
	\begin{tabular}{ |c | l|} 
		\hline
		\textbf{Symbol}&  \textbf{Description} \\ 
		\hline
		\hline
		$t$ & Time index\\
		\hline
		$W_{HL}$ & Switch cost from H to L\\
		\hline
		$W_{LH}$ & Switch cost from L to H\\
		\hline
		W       & Sum of switching costs ($W_{HL} + W_{LH})$ \\
		\hline
		$c_t$       & Rent cost per time-slot  $t$\\
		\hline
		$c_{L}$       &$c_t$ in state $S_L$\\
		\hline
		$c_{H}$       &$c_t$ in state $S_H$\\
		\hline
		
		$x_t$       & Request arrivals in time-slot $t$\\
		\hline
		$\underbar{\text{$x$}}_{t,L}$       & $\min \{x_t, \kappa_L\}$ \\
		\hline
		$\underbar{\text{$x$}}_{t,H}$       & $\min \{x_t, \kappa_H\}$ \\
		\hline
		$\delta_{t,L}$       & $x_t-\underbar{\text{$x$}}_{t,L}$\\
		\hline
		$\delta_{t,H}$       & $x_t-\underbar{\text{$x$}}_{t,H}$\\
		\hline
		$r^*(t)$   & Indicator variable; \textit{H} if the state is $S_H$ by  OPT-OFF during \\ & time-slot $t$ and \textit{L} otherwise \\
		\hline
		$r^{\text{BLTN}}(t)$   & Indicator variable; \textit{H} if the state is $S_H$ by  BLTN during time-slot $t$ \\ & and \textit{L} otherwise \\
		\hline
		$\eta$   & Notation for a policy\\
		\hline
		$C^{\eta}(n,m)$   & Total cost incurred by the policy $\eta$ in the interval $[n,m]$\\
		\hline
		$C^{\text{OPT-OFF}}(n,m)$  & Total cost incurred by the offline optimal  policy in the \\ & interval $[n,m]$\\
		\hline
		Frame $i$  & The interval between the $i^{\text{th}}$ and the $(i+1)^{\text{th}}$ switch to $S_H$ by \\ & the offline optimal policy\\
		\hline
		$C^{\text{OPT-OFF}}(i)$   & Total cost incurred by the offline optimal  policy in Frame $i$\\
		\hline
		$C^{\text{BLTN}}(i)$   & Total cost incurred by BLTN in Frame $i$\\
		\hline
	\end{tabular}}
	\vspace*{10pt}
	\caption{Notation}
	\vspace*{-10pt}
	\label{table:proofs}
\end{table}

We use the following lemmas to prove Theorem \ref{thm:BLTN_adv_online}(a). 

The following two lemmas give lower bounds of the difference in the number of requests served at the edge-server in the time interval between switching from $L \rightarrow H \rightarrow L$, and $H \rightarrow L \rightarrow H$ by OPT-OFF.
\begin{lemma}\label{lem:lemma_opt_H}
	
	If $r^*(n-1) = L$, $r^*(t) = H$ for $n \leq t \leq m$ and $r^*(m+1) = L$, then,  
	$
	\displaystyle\sum_{l=n}^m (\underbar{\text{$x$}}_{l,H} - \underbar{\text{$x$}}_{l,L})  \geq W +	\displaystyle\sum_{l=n}^m (c_H - c_L).
	$
\end{lemma}

\begin{proof}
	The cost incurred by OPT-OFF in $n \leq t \leq m + 1$ is $ W+\displaystyle\sum_{l=n}^m c_H + c_L +\displaystyle\sum_{l=n}^m \delta_{l,H} + \delta_{m+1,L}.$ We prove Lemma \ref{lem:lemma_opt_H} by contradiction. Let us assume that $\displaystyle\sum_{l=n}^{m} (\underbar{\text{$x$}}_{l,H} - \underbar{\text{$x$}}_{l,L}) < W +	\displaystyle\sum_{l=n}^{m} (c_H - c_L)$. We construct another policy $\eta$ which behaves same as OPT-OFF except that $r_{\eta}(t) = L$ for $n \leq t \leq m + 1$.
	The total cost incurred by $\eta$ in $n \leq t \leq m + 1$ is $\displaystyle\sum_{l=n}^{m+1} (\text{$x$}_{l} - \underbar{\text{$x$}}_{l, L}) + \displaystyle\sum_{l=n}^{m+1} c_L.$ It follows that
	$C^{\eta}(n,m+1)-C^{\text{OPT-OFF}}(n,m+1)= \displaystyle\sum_{l=n}^m (\underbar{\text{$x$}}_{l,H} - \underbar{\text{$x$}}_{l,L})-W_{LH} - W_{HL} -\displaystyle\sum_{l=n}^m (c_H - c_L),$ which is negative by our assumption. This contradicts the definition of the OPT-OFF policy, thus proving the result.
\end{proof}

\begin{lemma}\label{lem:lemma_opt_L}
	
	If $r^*(n-1) = H$, $r^*(t) = L$ for $n \leq t \leq m$ and $r^*(m+1) = H$, then,  
	$
	\displaystyle\sum_{l=n}^m (\underbar{\text{$x$}}_{l,H} - \underbar{\text{$x$}}_{l,L}) +  W \leq 	\displaystyle\sum_{l=n}^m (c_H - c_L).
	$
\end{lemma}

\begin{proof}
	The cost incurred by OPT-OFF in $n \leq t \leq m + 1$ is $ W+\displaystyle\sum_{l=n}^m c_L + c_H +\displaystyle\sum_{l=n}^m \delta_{l,L} + \delta_{m+1,H}.$ We prove Lemma \ref{lem:lemma_opt_L} by contradiction. Let us assume that $\displaystyle\sum_{l=n}^m (\underbar{\text{$x$}}_{l,H} - \underbar{\text{$x$}}_{l,L}) +  W > \displaystyle\sum_{l=n}^m (c_H - c_L)$. We construct another policy $\eta$ which behaves same as OPT-OFF except that $r_{\eta}(t) = H$ for $n \leq t \leq m$.
	The total cost incurred by $\eta$ in $n \leq t \leq m+1$ is $\displaystyle\sum_{l=n}^{m+1} (\text{$x$}_{l} - \underbar{\text{$x$}}_{l, H}) + \displaystyle\sum_{l=n}^{m+1} c_H.$ It follows that
	$C^{\eta}(n,m+1)-C^{\text{OPT-OFF}}(n,m+1)= \displaystyle\sum_{l=n}^m (\underbar{\text{$x$}}_{l,L} - \underbar{\text{$x$}}_{l,H})- W + \displaystyle\sum_{l=n}^m (c_H - c_L),$ which is negative by our assumption. This contradicts the definition of the OPT-OFF policy, thus proving the result.
\end{proof}

The next lemma shows that if the difference in the number of requests that can be served by the edge server in the two states in a time-interval exceeds a certain value (which is a function of the length of that time-interval) and the state is $S_L$ at the beginning of this time-interval, then OPT-OFF switches to state $S_H$ the service at least once in the time-interval.
\begin{lemma}\label{lem:OPT_download_H}
	If $r^*(n-1) = L$, and $\displaystyle\sum_{l=n}^m (\underbar{\text{$x$}}_{l,H} - \underbar{\text{$x$}}_{l,L})  \geq W +	\displaystyle\sum_{l=n}^m (c_H - c_L)$, then OPT-OFF switches the state to $S_H$ at least once in the interval from time-slots $n$ to $m$.
\end{lemma}

\begin{proof}
	We prove Lemma \ref{lem:OPT_download_H} by contradiction. We  construct another policy $\eta$ which behaves same as OPT-OFF except that $r_{\eta}(t) = H$ for $n \leq t \leq m$. 
	The total cost incurred by $\eta$ in $n \leq t \leq m$ is $C^{\eta}(n,m)=W +\displaystyle\sum_{l=n}^m c_H+\displaystyle\sum_{l=n}^m \delta_{l,H}$. It follows that
	$C^{\eta}(n,m)-C^{\text{OPT-OFF}}(n,m)= W +\displaystyle\sum_{l=n}^m (c_H - c_L)-\displaystyle\sum_{l=n}^m (\underbar{\text{$x$}}_{l,H} - \underbar{\text{$x$}}_{l,L}),$ which is negative. Hence there exists at least one policy $\eta$ which performs better than OPT-OFF. This contradicts the definition of the OPT-OFF policy, thus proving the result.
\end{proof}


The next lemma provides a lower bound on the duration for which OPT-OFF hosts  the service once it is fetched.
\begin{lemma}\label{lem:OPT_slots_1}
	Once OPT-OFF switches to state H,  the state is constant for at least $\frac{ W}{(\kappa_H - \kappa_L) -(c_H - c_L)}$ slots.
\end{lemma}

\begin{proof}
	Suppose OPT-OFF switches to $S_H$ at the end of the $(n-1)^{\text{th}}$ time-slot and switches to $S_L$ at the end of time-slot $m > n$. From Lemma \ref{lem:lemma_opt_H}, $\displaystyle\sum_{l=n}^m (\underbar{\text{$x$}}_{l,H} - \underbar{\text{$x$}}_{l,L}) \geq W +\displaystyle\sum_{l=n}^m (c_H - c_L)$. Since $(\underbar{\text{$x$}}_{l,H} - \underbar{\text{$x$}}_{l,L}) \leq (m-n+1)\times(\kappa_H - \kappa_L)$ and $\displaystyle\sum_{l=n}^m (c_H - c_L) \geq (m-n+1)\times(c_H - c_L)$, $(m-n+1)\times(\kappa_H - \kappa_L) \geq W+(m-n+1)\times(c_H - c_L)$, i.e, $(m-n+1) \geq \frac{W}{(\kappa_H - \kappa_L)-(c_H - c_L)}$. This proves the result.
\end{proof}

\begin{lemma}\label{lem:OPT_slots_2}
	Once OPT-OFF switches to state L, the state is constant for at least $\frac{ W}{(c_H - c_L)}$ slots.
\end{lemma}

\begin{proof}
	Suppose OPT-OFF switches to $S_L$ at the end of the $(n-1)^{\text{th}}$ time-slot and switches to $S_L$ at the end of time-slot $m > n$. From Lemma \ref{lem:lemma_opt_L}, $\displaystyle\sum_{l=n}^m (\underbar{\text{$x$}}_{l,H} - \underbar{\text{$x$}}_{l,L}) + W \leq \displaystyle\sum_{l=n}^m (c_H - c_L)$. Since $(\underbar{\text{$x$}}_{l,H} - \underbar{\text{$x$}}_{l,L}) \geq 0$ and $\displaystyle\sum_{l=n}^m (c_H - c_L) = (m-n+1)\times(c_H - c_L)$, $W \leq (m-n+1)\times(c_H - c_L)$, i.e, $(m-n+1) \geq \frac{W}{(c_H - c_L)}$. This proves the result.
\end{proof}

The next lemma gives an upper bound on the difference in the number of requests that can be served by the edge server (between the two states subject to its computation power constraints) in a time-interval such that BLTN is in state $S_L$ during the time-interval and fetches it in the last time-slot of the time-interval.

\begin{lemma}\label{lem:max_requests_L}
	Let $r^{\text{BLTN}}(n-1)=H$, $r^{\text{BLTN}}(t)=L$ for $n \leq t \leq m$ and $r^{\text{BLTN}}(m+1)=H$. Then for any $n\leq n'< m$,  $\displaystyle\sum_{l=n'}^m \left(\underbar{\text{$x$}}_{l,H} - \underbar{\text{$x$}}_{l,L}\right) < \displaystyle\sum_{l=n'}^{m-1} (c_H - c_L)+W + (\kappa_H - \kappa_L).$   
\end{lemma}
\begin{proof}
	Given $r^{\text{BLTN}}(m)=L$ and $r^{\text{BLTN}}(m+1)=H$, then for any $n\leq n'< m$, $\displaystyle\sum_{l=n'}^{m-1} (\underbar{\text{$\delta$}}_{l,L} - \underbar{\text{$\delta$}}_{l,H}) <  \displaystyle\sum_{l=n'}^{m-1} (c_H - c_L)+W.$ 
	By definition,
	$
	\displaystyle\sum_{l=n'}^m \left(\underbar{\text{$x$}}_{l,H} - \underbar{\text{$x$}}_{l,L}\right) = \left(\underbar{\text{$x$}}_{l,H} - \underbar{\text{$x$}}_{l,L}\right)+\underbar{\text{$x$}}_{m,H} - \underbar{\text{$x$}}_{m,L}< \displaystyle\sum_{l=n'}^{m-1} (c_H - c_L)+W +(\kappa_H - \kappa_L),
	$
	thus proving the result.	
\end{proof}

\begin{lemma}\label{lem:max_requests_H}
	Let $r^{\text{BLTN}}(n-1)=L$, $r^{\text{BLTN}}(t)=H$ for $n \leq t \leq m$ and $r^{\text{BLTN}}(m+1)=L$. Then for any $n\leq n'< m$,  $\displaystyle\sum_{l=n'}^m \left(\underbar{\text{$x$}}_{l,H} - \underbar{\text{$x$}}_{l,L}\right) < \displaystyle\sum_{l=n'}^{m-1} (c_H - c_L)+W .$   
\end{lemma}
\begin{proof}
	Given $r^{\text{BLTN}}(m)=H$ and $r^{\text{BLTN}}(m+1)=L$, then for any $n\leq n'< m$, $\displaystyle\sum_{l=n'}^{m-1} (\underbar{\text{$\delta$}}_{l,H} - \underbar{\text{$\delta$}}_{l,L}) + \displaystyle\sum_{l=n'}^{m-1} (c_H - c_L) < W.$ 
	By definition,
	$
	\displaystyle\sum_{l=n'}^m \left(\underbar{\text{$x$}}_{l,H} - \underbar{\text{$x$}}_{l,L}\right) = \left(\underbar{\text{$x$}}_{l,H} - \underbar{\text{$x$}}_{l,L}\right)+\underbar{\text{$x$}}_{m,H} - \underbar{\text{$x$}}_{m,L} > \displaystyle\sum_{l=n'}^{m-1} (c_H - c_L) - W,
	$
	thus proving the result.	
\end{proof}

Consider the event where both BLTN and OPT-OFF have hosted in state $S_H$ in a particular time-slot. The next lemma states that given this, OPT-OFF switches states to $S_L$ before BLTN. 

\begin{lemma}\label{lem:eviction_H}
	If $r^{\text{BLTN}}(n)=H$, $r^*(t) = H$ for $n \leq t \leq m$, and $r^*(m+1) = L$. Then,  $r^{\text{BLTN}}(t)=H$ for $n+1 \leq t \leq m+1$.
\end{lemma}

\begin{proof}
	We prove this by contradiction. Let $\exists \widetilde{m}<m$ such that $r^{\text{BLTN}}(\widetilde{m}+1)=L$. Then, from Algorithm \ref{algo:BLTN}, there exists an integer $\tau>0$  such that  $\displaystyle\sum_{l=\widetilde{m}-\tau+1}^{\widetilde{m}} (\underbar{\text{$x$}}_{l,H} - \underbar{\text{$x$}}_{l,L})< \displaystyle\sum_{l=\widetilde{m}-\tau+1}^{\widetilde{m}} (c_H - c_L)-W.$ 	The cost incurred by OPT-OFF in the interval $\widetilde{m}-\tau+1$ to $\widetilde{m}$ is $\displaystyle\sum_{l=\widetilde{m}-\tau+1}^{\widetilde{m}} c_H+\displaystyle\sum_{l=\widetilde{m}-\tau+1}^{\widetilde{m}}\delta_{l,H}$. 
	
	Consider an alternative policy $\eta$ for which $r_{\eta}(t)=0$ for $\widetilde{m}-\tau+1 \leq t \leq \widetilde{m}$, $r_{\eta}(\widetilde{m}+1)=1$, and $r_{\eta}(t)=r^*(t)$ otherwise. It follows that
	$C^{\eta}-C^{\text{OPT-OFF}}= \displaystyle\sum_{l=\widetilde{m}-\tau+1}^{\widetilde{m}} (\underbar{\text{$x$}}_{l,H} - \underbar{\text{$x$}}_{l,L}) + W - \displaystyle\sum_{l=\widetilde{m}-\tau+1}^{\widetilde{m}} (c_H - c_L)$ which is negative by our assumption. This contradicts the definition of the OPT-OFF policy, thus proving the result.	
\end{proof}

Consider the case where both BLTN and OPT-OFF have hosted in state $S_H$ in a particular time-slot. From the previous lemma, we know that, OPT-OFF switches states to $S_L$ before BLTN. The next lemma gives a lower bound on the difference in the number of requests that can be served by the edge server in the interval which starts when OPT-OFF switches states to $S_L$ from the edge server  and ends when BLTN switches states to $S_L$ from the edge server. 

\begin{lemma}\label{lem:min_requests_1}
	Let $r^*(n-1)= H, \ r^*(n)=L, \ r^{\text{BLTN}}(t) = H$ for $n-1 \leq t \leq m$ and  $r^{\text{BLTN}}(m+1)=L$. Then for any $n\leq n'< m$, $\displaystyle\sum_{l=n'}^m (\underbar{\text{$x$}}_{l,H} - \underbar{\text{$x$}}_{l,L}) \geq 	\displaystyle\sum_{l=n'}^{m-1} (c_H - c_L) - W_{HL}.$    	
\end{lemma}

\begin{proof}
	Given $r^{\text{BLTN}}(m)=H$ and $r^{\text{BLTN}}(m+1)=L$, then  for any $n\leq n'< m$, $\displaystyle\sum_{l=n'}^{m-1} (\underbar{\text{$x$}}_{l,H} - \underbar{\text{$x$}}_{l,L}) > 	\displaystyle\sum_{l=n'}^{m-1} \Delta c - W$. 
	By definition,
	$
	\displaystyle\sum_{l=n'}^m (\underbar{\text{$x$}}_{l,H} - \underbar{\text{$x$}}_{l,L}) = \left(\displaystyle\sum_{l=n'}^{m-1} (\underbar{\text{$x$}}_{l,H} - \underbar{\text{$x$}}_{l,L}) \right)+(\underbar{\text{$x$}}_{l,H} - \underbar{\text{$x$}}_{l,L}) > \displaystyle\sum_{l=n'}^{m-1} \Delta c - W+0,
	$
	thus proving the result. 	
\end{proof}

Our next result states that BLTN does not switch states to $S_H$ in the interval between a switch to $S_L$ and the subsequent switch to $S_H$ by OPT-OFF.

\begin{lemma}\label{lem:RR_no_download} If $r^*(n-1) = H, \ r^*(t) = L$ for $n \leq t \leq m$, and $r^*(m+1)=H$, then BLTN does not switch to state H in time-slots $n, n+1, \cdots, m-1$. 
\end{lemma}

\begin{proof}
	We prove this by contradiction. Let BLTN switch states to $S_H$ in time-slot $t$ where $n\leq t \leq m-1.$ 
	Then from Algorithm \ref{algo:BLTN}, there exists an integer $\tau>0$ such that $t-\tau\geq n$ and $\displaystyle\sum_{l=t-\tau+1}^t (\underbar{\text{$x$}}_{l,H} - \underbar{\text{$x$}}_{l,L}) \geq \displaystyle\sum_{l=t-\tau+1}^t (c_H - c_L) + W_{LH}  +W_{HL}.$
	If this condition is true, by Lemma \ref{lem:OPT_download_H},  OPT-OFF would have fetched the service at least once in the interval $t-\tau+1$ and $t$ for all $n\leq t \leq m-1.$ Hence BLTN does not switch states to $S_H$ between $n$ and $m-1.$
\end{proof}

The next lemma states that in the interval between a switch from $L \rightarrow H$ and subsequent switch from $H \rightarrow L$ by OPT-OFF, BLTN hosts the service for at least one time-slot. 

\begin{lemma}\label{lem:RR_one_download}
	If $r^*(n-1) = L$, $r^*(t) = H$ for $n \leq t \leq m$ and $r^*(m+1) = L$, then,
	for some  $n < t \leq m$, $r^{\text{BLTN}}(t) = H.$
\end{lemma}
\begin{proof}
	We prove this by contradiction. Let $r^{\text{BLTN}}(t)=L$ for all $n \leq t \leq m$. Then by the definition of the BLTN policy,  $\displaystyle\sum_{l=t-\tau+1}^t (\underbar{\text{$x$}}_{l,H} - \underbar{\text{$x$}}_{l,L}) < \displaystyle\sum_{l=t-\tau+1}^t (c_H - c_L) + W$ for any $\tau>0$ and $\tau \leq t-n+1$.
	If we choose $t=m$   then  $\displaystyle\sum_{l=n}^m (\underbar{\text{$x$}}_{l,H} - \underbar{\text{$x$}}_{l,L}) < \displaystyle\sum_{l=n}^m (c_H - c_L) + W_{HL}  + W_{LH}$, which is false from Lemma \ref{lem:lemma_opt_H}. This contradicts our assumption.  
\end{proof}

If both BLTN and OPT-OFF are in state $S_H$ in a particular time-slot, from Lemma \ref{lem:eviction_H}, we know that OPT-OFF switches states to $S_L$ before BLTN. The next lemma states that BLTN switches states to $S_L$ before the next time OPT-OFF switches to state $S_H$.

\begin{lemma}\label{lem:RR_one_L}
	If $r^*(n-1) = H$, $r^*(t) = L$ for $n \leq t \leq m$, $r^*(m+1) = H$, and $r^{\text{BLTN}}(n-1)=H$, then, BLTN switches states to $S_L$ by the end of time-slot $m$ and $r^{\text{BLTN}}(m+1) = L$.
\end{lemma}
\begin{proof}
	We prove this by contradiction. Assume that BLTN does not switch states to $S_L$ in any time slot $t$ for all  $n \leq t \leq m$. 
	Then from the definition of the BLTN policy, $\displaystyle\sum_{l=t-\tau+1}^{t} (\underbar{\text{$x$}}_{l,H} - \underbar{\text{$x$}}_{l,L}) + W \geq 	\displaystyle\sum_{l=t-\tau+1}^{t} (c_H - c_L)$ for all $\tau$ such that $0 < \tau \leq t-n+1$.
	As a result, at $t=m$, $\displaystyle\sum_{l=n}^m (\underbar{\text{$x$}}_{l,H} - \underbar{\text{$x$}}_{l,L}) + W_{HL}  + W_{LH} > \displaystyle\sum_{l=n}^m (c_H - c_L)$. 
		Given this, it follows that OPT-OFF will not switch states to $S_L$ at the end of time-slot $n-1$. 
	This contradicts our assumption.
	By Lemma \ref{lem:RR_no_download}, BLTN does not switch states to $S_H$ in the interval between switches from $S_H$ to $S_L$ and back by OPT-OFF. Therefore, $r^{\text{BLTN}}(m+1) = 0$.
\end{proof}

To compare the costs incurred by BLTN and OPT-OFF we divide time into frames $[1,t_1-1]$, $[t_1,t_2-1], [t_2,t_3-1],\cdots,$ where $t_i-1$ is the time-slot in which OPT-OFF switches to state $S_H$ for the $i^{\text{th}}$ time for $i\in\{1,2,\cdots\}.$ Our next result characterizes the sequence of events that occur in any such frame.

\begin{lemma}
	\label{lem:frameStructure}
	Consider the interval $[t_i, t_{i+1}-1]$ such that OPT-OFF switches the state to $S_H$ at the end of time-slot $t_{i}-1$ and switches the state to $S_H$ again the end of time-slot $t_{i+1}-1$. By definition, there exists $\tau \in [t_i, t_{i+1}-2]$ such that OPT-OFF switches the state to $S_L$ in time-slot $\tau$. 
	BLTN switches to $S_H$ from $S_L$ and back exactly once each in $[t_1, t_2-1]$. The switch to $S_H$ by BLTN is in time-slot $t_{LH}^{\text{BLTN}} $ such that $t_1 \leq t_{LH}^{\text{BLTN}} \leq \tau$ and
	the switch to $S_L$ by BLTN is in time-slot $t_{HL}^{\text{BLTN}}$ such that $\tau < t_{HL}^{\text{BLTN}} < t_2$ (Figure \ref{fig:OPT_RR_frame}).
\end{lemma}

\begin{figure}[ht]
	\centering
	\begin{tikzpicture}
	\foreach \x in {0,0.3,0.6,...,7.2}{
		\draw[] (\x,0) --  (\x+0.3,0);
		\draw[gray] (\x,-1mm) -- (\x,1mm);
	}
	
	\draw[gray] (7.2,-1mm) -- (7.2,1mm);
	\draw[-latex] (0.15,10mm) -- node[above=5mm]{$t_{i}-1$} (0.15,0mm);
	\draw[-latex] (6.75,10mm) -- node[above=5mm]{$t_{i+1}-1$} (6.75,0mm);
	\draw[-latex] (4.05,0mm) -- node[above=5mm]{$\tau$} (4.05,10mm);
	\draw[-latex,color=red] (2.25,10mm) -- node[above=5mm]{$t_{LH}^{\text{BLTN}}$} (2.25,0mm);
	\draw[-latex,color=red] (5.55,0mm) -- node[above=5mm]{$t_{HL}^{\text{BLTN}}$} (5.55,10mm);
	
	\draw[gray] (7.2,-1mm) -- (7.2,1mm);
	\draw[-latex] (0.15,10mm) -- node[above=5mm]{} (0.15,0mm);
	\draw[-latex] (6.75,10mm) -- node[above=5mm]{} (6.75,0mm);
	\draw[-latex] (4.05,0mm) -- node[above=5mm]{} (4.05,10mm);
	\draw[-latex,color=red] (2.25,10mm) -- node[above=5mm]{} (2.25,0mm);
	\draw[-latex,color=red] (5.55,0mm) -- node[above=5mm]{} (5.55,10mm);
	
	\draw[<->,color=black] (0.3,-2mm) -- node[below=0mm,pos=0.5]{Frame $i$} (6.9,-2mm);
	
	\draw[black] (0.3,-1mm) -- (0.3,-3mm);
	\draw[black] (6.9,-1mm) -- (6.9,-3mm);
	
	\node[] at (7.7,-0.9) {OPT-OFF};
	
	\filldraw[fill=black!] (0.25,-1) rectangle (4.15,-0.7);
	\filldraw[fill=white!40] (4.15,-1) rectangle (6.85,-0.7);
	\node[] at (7.25,-1.4) {BLTN};
	\filldraw[fill=white!40] (0.25,-1.5) rectangle (2.4,-1.2);
	\filldraw[fill=red!] (2.4,-1.5) rectangle (5.7,-1.2);
	\filldraw[fill=white!40] (5.7,-1.5) rectangle (6.85,-1.2);

	\end{tikzpicture}
	\caption{Illustration of Lemma \ref{lem:frameStructure} showing switches between $S_H$ and $S_L$ by OPT-OFF and BLTN in the $i^{\text{th}}$ frame. Downward arrows represent switch to $S_H$, upward arrows indicate switch to $S_L$. Black and red arrows correspond to the OPT-OFF and BLTN policy respectively. The two bars below the timeline indicate the state of the edge server  under OPT-OFF and BLTN. The solid black and solid red portions represent the intervals during with OPT-OFF and BLTN host  the service on the edge server respectively}
	\label{fig:OPT_RR_frame}
\end{figure}
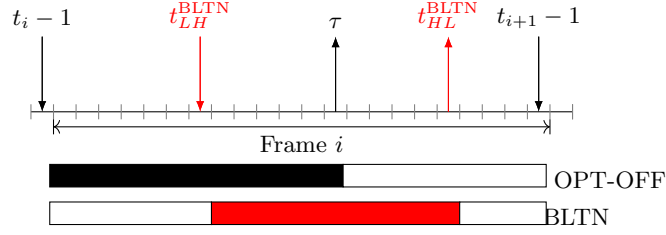

\begin{proof}
	Without loss of generality, we prove the result for $i=1$.
	Since $r^*(t_1-1) = L$, $r^*(t) = H$ for $t_1 \leq t \leq \tau$ and $r^*( \tau) = L$ then by Lemma \ref{lem:RR_one_download}, $r^{\text{BLTN}}(t_{LH}^{\text{BLTN}}) = H$ for some  $t_1 < t_{LH}^{\text{BLTN}} \leq \tau.$ In addition, by Lemma \ref{lem:RR_one_L}, $r^{\text{BLTN}}(t_1) = L$. Therefore, BLTN the state to $S_H$ at least once in the interval $[t_1, t_2-1]$. 
		By Lemma \ref{lem:eviction_H}, if $t_{LH}^{\text{BLTN}} < \tau$, since both BLTN and OPT-OFF stay in state $S_H$ during time-slot $t_{LH}^{\text{BLTN}}+1$, OPT-OFF switches to $S_L$ before BLTN, therefore, once fetched, BLTN does not switch to $S_L$ before time-slot $\tau+1$, i.e., $r^{\text{BLTN}}(t)=H$ for $t_{LH}^{\text{BLTN}}+1 \leq t \leq \tau+1.$
		Since $r^*(\tau) = H$, $r^*(t) = L$ for $\tau+1 \leq t \leq t_2-1$ and $r^*(t_2) = H$, then by Lemma \ref{lem:RR_one_L}, BLTN switches states to $S_L$ in time-slot $t_{HL}^{\text{BLTN}}$ such that $\tau < t_{HL}^{\text{BLTN}} \leq t_2-1$.
	In addition, once switched to $S_L$ at $t_{HL}^{\text{BLTN}} \leq t_2-1$, BLTN does not switch to $S_H$ again in the before time-slot $t_2$ by Lemma \ref{lem:RR_no_download}. 
		This completes the proof.
\end{proof}

\begin{proof}[of Theorem \ref{thm:BLTN_adv_online}(a)] As mentioned above, to compare the costs incurred by BLTN and OPT-OFF we divide times into frames $[1,t_1-1]$, $[t_1,t_2-1], [t_2,t_3-1],\ldots,$ where $t_i-1$ is the time-slot in which OPT-OFF downloads the service for the $i^{\text{th}}$ time for $i\in\{1,2,\ldots, k\}.$

	For convenience, we account for the switch costs incurred by OPT-OFF in time-slot $t_i$ in the cost incurred by OPT-OFF in Frame $i$. Given this, the cost under BLTN and OPT-OFF is the same for $[1,t_1-1]$ (Frame 0) since both policies host  the service in the same state. 
	
	Note that if the total number of switches made by OPT-OFF is less than $2k<\infty$, there are exactly $k+1$ frames (including Frame 0). The $(k+1)^{\text{th}}$ frame either has no switch by OPT-OFF to $S_L$ or OPT-OFF switches to $S_L$ and then never switches back. 
	
	We now focus on Frame $i$, such that $0<i<k$, where $2k$ is the total number of switches made by OPT-OFF. 
	
	Note: $\Delta \kappa = \kappa_H - \kappa_L, \Delta c = c_H - c_L$
	
	Without loss of generality, we focus on Frame 1. Recall the definitions of $\tau$, $t_{HL}^{\text{BLTN}}$, and $t_{LH}^{\text{BLTN}}$ from Lemma \ref{lem:frameStructure}, also seen in Figure \ref{fig:OPT_RR_frame}. By Lemma \ref{lem:frameStructure}, we have that BLTN switches from $S_H$ to $S_L$ and back exactly once each in $[t_1, t_2-1]$ such that the switch to $S_H$ by BLTN is in time-slot $t_{LH}^{\text{BLTN}} $ such that $t_1 \leq t_{LH}^{\text{BLTN}} \leq \tau$ and the eviction by BLTN is in time-slot $t_{HL}^{\text{BLTN}}$ such that $\tau < t_{HL}^{\text{BLTN}} < t_2$.
	
	Both OPT-OFF and BLTN makes one pair of switches in the frame. Hence the difference in the switch costs is zero. We now focus on the service and rent cost incurred by the two policies. 
	
	By Lemma \ref{lem:max_requests_L}, the difference in total cost $= C^{BLTN} - C^{OPT-OFF} = \displaystyle\sum_{l=t_1}^{t_{LH}^{\text{BLTN}}} (\underbar{\text{$\delta$}}_{l,L} - \underbar{\text{$\delta$}}_{l,H}) + \displaystyle\sum_{l=t_1}^{t_{LH}^{\text{BLTN}}} (c_L - c_H) <  \displaystyle\sum_{l=t_1}^{t_{LH}^{\text{BLTN}}-1} (c_H - c_L)+W +(\kappa_H - \kappa_L) + \displaystyle\sum_{l=t_1}^{t_{LH}^{\text{BLTN}}} (c_L - c_H)  = W + \Delta \kappa - \Delta c$.
    
		The service and rent cost incurred by OPT-OFF and BLTN in $[t_{f}^{\text{BLTN}}+1, \tau]$ are equal. 
		
	By Lemma \ref{lem:max_requests_H}, the difference of costs in $[\tau+1, t_{HL}^{\text{BLTN}}]$ is $ C^{BLTN} - C^{OPT-OFF} = \displaystyle\sum_{l=\tau+1}^{t_{HL}^{\text{BLTN}}} (\underbar{\text{$\delta$}}_{l,H} - \underbar{\text{$\delta$}}_{l,L}) + \displaystyle\sum_{l=\tau+1}^{t_{HL}^{\text{BLTN}}} (c_H - c_L) <  \displaystyle\sum_{l=\tau+1}^{t_{HL}^{\text{BLTN}}-1} (c_L - c_H)+W + \displaystyle\sum_{l=\tau+1}^{t_{HL}^{\text{BLTN}}} (c_H - c_L)  = W + \Delta c$.
	
	The service and rent cost incurred by OPT-OFF and BLTN in $[t_{e}^{\text{BLTN}}+1, t_2-1]$ are equal. 
	
	Let $C^{\text{BLTN}}(i),$ $C^{\text{OPT-OFF}}(i)$ denote the costs incurred in the $i^{\text{th}}$ frame by BLTN and OPT-OFF respectively.
	We therefore have that,
	\begin{align}
	C^{\text{BLTN}}(i)-C^{\text{OPT-OFF}}(i)\leq  2W + \kappa_H - \kappa_L
	\label{ineq:bound1_RR}
	\end{align}
	
	By Lemma \ref{lem:OPT_slots_1}, once OPT-OFF switches to state H, it stays put  for at least $\tau_H = \frac{ W}{\Delta \kappa-\Delta c}$ slots. And by Lemma \ref{lem:OPT_slots_2}, once OPT-OFF switches to state L, it stays put  for at least $\tau_L = \frac{ W}{\Delta \kappa}$ slots.
	Therefore,
	\begin{align}
	&C^{\text{OPT-OFF}}(i)\geq W_{LH}+ c_H \tau_H + W_{HL} + c_L \tau_L  \nonumber
	\label{ineq:bound1_OPT}
	\end{align}
	
	From \eqref{ineq:bound1_RR} and \eqref{ineq:bound1_OPT}, 
	\begin{align}
	C^{\text{BLTN}}(i) \leq  \left(1 + \frac{ 2W + \Delta \kappa}{ W \left( 1 + \frac{c_H}{\Delta \kappa - \Delta c}+ \frac{c_L}{\Delta \kappa } \right)}\right)C^{\text{OPT-OFF}}(i).
	\end{align}

	Frames $1$ to $k-1$ have now been characterized completely.
	
	For Frame $k$, which is the last frame, there are two possible cases, one where OPT-OFF switches states to $S_L$ in Frame $k$, in which case the analysis for Frame $k$ is identical to that of Frame $1$, and the other when OPT-OFF does not switch states to $S_L$ in Frame $k$. We now focus on the latter.

	Given that OPT-OFF switches to state $S_H$ the service in time-slot $t_{k}-1$, there exists $m>t_{k}$ such that
	$\displaystyle\sum_{l=t_k}^m (\underbar{\text{$x$}}_{l,H} - \underbar{\text{$x$}}_{l,L}) \geq W +\displaystyle\sum_{l=t_k}^m (c_H - c_L)$. 
	By Step~8 in Algorithm \ref{algo:BLTN}, BLTN switches to state $S_H$ at the end of
	time-slot $m$.
	Let $\tau_k = m-t_k$. By Lemma \ref{lem:max_requests_H}, the difference in the number of requests that can be served by the edge server during these $\tau_k$ time-slots is at most $\displaystyle\sum_{l=n'}^{m-1} (c_H - c_L)+W.$
	Since BLTN is in state $S_L$ during these  $\tau_k$ time-slots, the rent cost incurred by BLTN is $c_L$ per slot and the service cost incurred by BLTN is at most $W+\displaystyle\sum_{l=t_k}^{m-1} c_L+\kappa_L + \displaystyle \sum_{l = t_k}^{m} \delta_{l,L}$. OPT-OFF rents the edge server during these $\tau_k$ time-slots at cost $\displaystyle \sum_{l = t_k}^{m} c_H$ and the service cost incurred by OPT-OFF is $\displaystyle \sum_{l = t_k}^{m - 1} \delta_{l,H}$. There is no difference between the cost of BLTN and OPT-OFF after the first $\tau_k$ slots in Frame $k$. It follows that
	\begin{align}
	C^{\text{BLTN}}(k)-C^{\text{OPT-OFF}}(k)\leq 2W+\Delta \kappa.
	\label{ineq:bound1_RRk}
	\end{align}
	From \eqref{ineq:bound1_OPT} and \eqref{ineq:bound1_RRk},

	\begin{align}
	C^{\text{BLTN}}(k) &\leq \left(1 + \frac{ 2W + \Delta \kappa}{ W \left( 1 + \frac{c_H}{\Delta \kappa - \Delta c}+ \frac{c_L}{\Delta \kappa } \right)}\right)C^{\text{OPT-OFF}}(k).
	\end{align}

	Stitching together the results obtained for all frames, the result follows.  
	
\end{proof}

\label{sec:proofthm1b}

\subsection{Proof of Theorem \ref{thm:BLTN_adv_online}(b)}

\begin{proof}
	Let $\mathcal{P}$ be a given deterministic online policy and $C^\mathcal{P}(a)$ be the cost incurred by this policy for the request sequence $a$.
	
	We first consider the case where $\mathcal{P}$ starts in state $S_H$. 
	
	We define $t^{(1)} \geq 1$ as the first time the policy $\mathcal{P}$ switches to state $S_H$ when there are $\kappa_H$ arrivals in each of the first $t^{(1)}$ time-slots. As $\mathcal{P}$ is a deterministic policy, the value of $t^{(1)}$ can be computed a-priori.
	
	We define $t^{(2)} \geq 1$ as the first time the policy $\mathcal{P}$ switches to state $S_L$ when there are $\kappa_L$ arrivals in each of the first $t^{(2)}$ time-slots. Similarly, as $\mathcal{P}$ is a deterministic policy, the value of $t^{(2)}$ can be computed a-priori.
	
	Consider the arrival process $a$ with $\kappa_H$ request arrivals each in the first $t^{(1)}$ time-slots and $\kappa_L$ request arrivals each in the next $t^{(2)}$ time-slots. 
	
	Note: $\Delta \kappa = \kappa_H - \kappa_L, \Delta c = c_H - c_L, W = W $
	
	It follows that $C^{\mathcal{P}}(a) = t^{(1)}(\Delta \kappa + c_L) + t^{(2)}c_H + W.$
	
	Consider an alternative policy ALT which is in state $S_H$ from time-slots $1$ to $t^{(1)}$ and in state $S_L$ from time-slots $t^{(1)} + 1$ to $t^{(1)} + t^{(2)}$. It follows that $C^{\text{ALT}}(a) = c_Ht^{(1)} + c_Lt^{(2)} + W.$ By definition,
	$
	\rho^{\mathcal{P}} \geq \dfrac{(\Delta \kappa + c_L)t^{(1)} + c_Ht^{(2)} + W}{c_Ht^{(1)} + c_Lt^{(2)} + W}.
	$
	Therefore,
	\[
	\rho^{\mathcal{P}} \geq
	\text{min } \bigg\{ \frac{\Delta \kappa + c_L}{c_H}, \frac{c_H}{c_L}, \frac{(\Delta \kappa + c_L) + c_H + W}{c_H + c_L + W} \bigg\}
	\]
	
\end{proof}

\label{sec:proofthm2}

\subsection{Proof of Theorem 2}

We use the following lemmas to prove Theorem \ref{thm:BLTN_stochastic_theorem}. 

\begin{lemma}
	\label{lemma:optimal_causal}
	Let $X_t$ be the number  of requests arriving in time-slot $t$, $\nu = \mathbb{E}[X_t]$,  $\underbar{$X$}_{t,H}=\min\{X_t,\kappa_H\}$ and $\mu_L =  \mathbb{E}[\underbar{$X$}_{t,L}]$, $\underbar{$X$}_{t,L}=\min\{X_t,\kappa_L\}$ and $\mu_L =  \mathbb{E}[\underbar{$X$}_{t,L}]$. 
	Let the rent cost per time-slot be $c_H$ or $c_L$ depending on the states $S_H$ or $S_L$ respectively.
	Under Assumption \ref{assum_stochastic}, let $\mathbb{E}[ C_t^{\text{OPT-ON}}]$ be the cost per time-slot incurred by the OPT-ON policy. Then,
	$
	\mathbb{E}[ C_t^{\text{OPT-ON}}] \geq \min\{c_H + \nu - \mu_H, c_L + \nu - \mu_L \}.
$
\end{lemma}
\begin{proof}
    If the service is hosted in state $S_H$ on the edge in time-slot $t$, the expected cost incurred is at least $\mathbb{E}[X_t - \min\{X(t) , \kappa_H\} + c_H]$ = $c_H + \nu - \mu_H$.
    Else, if the service is hosted in state $S_L$ on the edge in time-slot $t$, the expected cost incurred is at least $\mathbb{E}[X_t - \min\{X(t) , \kappa_L\} + c_L]$ = $c_L + \nu - \mu_L$.
\end{proof}

\begin{lemma}\label{lem:Hoeffding}
	Let $X_t$ be the number  of requests arriving in time-slot $t$, $\underbar{X}_{t,H}=\min\{X_t,\kappa_H\}$, $\mu_H =  \mathbb{E}[\underbar{X}_{t,H}]$, $\underbar{X}_{t,L}=\min\{X_t,\kappa_L\}$ and $\mu_L =  \mathbb{E}[\underbar{X}_{t,L}]$. 
	Let the rent cost per time-slot be $c_H$ or $c_L$ depending on the states $S_H$ or $S_L$ respectively. Define $\Delta \underbar{X}_l = \underbar{X}_{l,H} - \underbar{X}_{l,L}$, $\Delta \mu = \mu_H - \mu_L$, $\Delta \kappa = \kappa_H - \kappa_L$,  $\Delta c = c_H - c_L$,  $Y_l=\Delta \underbar{X}_l- \Delta c$, and  $Y=\sum\limits_{l=t-\tau+1}^t Y_l$ then $Y$ satisfies, for $(\Delta c- \Delta \mu) \tau + W > 0$,
	$
	\mathbb{P}\left(Y \geq W\right)  \leq \exp\left(-2\frac{((\Delta c- \Delta \mu)\tau+W)^2}{\tau (\Delta \kappa)^2}\right), 
	$
	and for $(\Delta \mu- \Delta c) \tau + W > 0$,
	$
	\mathbb{P}\left(Y \leq \tau c-W\right) \leq \exp\left(-2\frac{((\Delta \mu- \Delta c)\tau+W)^2}{\tau (\Delta \kappa)^2}\right). 
	$
\end{lemma}
\begin{proof}   Using i.i.d. condition of $\{X_t\}_{t\geq 1}$, it follows that for $s>0$, $\mathbb{E}[\exp(sY)]\leq \prod\limits_{l=t-\tau+1}^t \mathbb{E}[\exp(sY_l)]$.  Moreover, $Y_l\in [ -\Delta c, \Delta \kappa- \Delta c]$. Then the result  follows by Hoeffding's inequality. 
\end{proof}
\vspace{-1pt}
\begin{proof}[of Lemma \ref{lemma:difference_BLTNstochastic}] We first consider the case when $\Delta \mu> \Delta c$. We define the following events

$E_{t_1,t_2}: \displaystyle\sum_{l=t_1}^{t_2} \Delta \underbar{$X$}_l \leq \displaystyle\sum_{l=t_1}^{t_2} \Delta c - W$, $E^{\tau} =\displaystyle\bigcup_{t_1=1}^{\tau} E_{t_1,\tau}$, $E_{t-1} =\displaystyle\bigcup_{\tau=t-\lceil\frac{\lambda W}{\Delta \mu- \Delta c}\rceil}^{t-1} E^{\tau}$, $E_{t} =\displaystyle\bigcup_{\tau=t-\lceil\frac{\lambda W}{\Delta \mu- \Delta c}\rceil + 1}^{t} E^{\tau}, F: \displaystyle\sum_{l=t-\lceil\frac{\lambda W}{\Delta \mu- \Delta c}\rceil}^{t-1} \Delta \underbar{$X$}_l \geq \displaystyle\sum_{l=t-\lceil\frac{\lambda W}{\Delta \mu- \Delta c}\rceil}^{t-1} \Delta c+W$.  \\
	
	By Lemma \ref{lem:Hoeffding}, it follows that 
	$\mathbb{P}(E_{t_1,t_2}) \leq \exp\left(-2\frac{(\Delta \mu- \Delta c)^2 (t_2-t_1+1)}{(\Delta \kappa)^2}\right),$ 
	and therefore,
	\begin{align}
	\mathbb{P}(E^{\tau}) &\leq \displaystyle \sum_{t_1=1}^{\tau-\lceil\frac{W}{\Delta c}\rceil+1} \exp\left(-2\frac{(\Delta \mu- \Delta c)^2 (\tau-t_1+1)}{ (\Delta \kappa)^2}\right) \nonumber \\ &\leq \frac{\exp\left(-2\frac{(\Delta \mu- \Delta c)^2 \frac{W}{\Delta c}}{ (\Delta \kappa)^2}\right)}{1-\exp\left(-2\frac{(\Delta \mu- \Delta c)^2}{ (\Delta \kappa)^2}\right)}. \label{eq:E_tau}
	\end{align}
	Using \ref{eq:E_tau} and the union bound, $\mathbb{P}(E_{t-1})$ and $\mathbb{P}(E_t)$ are upper bounded by
	\begin{align}
	\bigg\lceil\frac{\lambda W}{\Delta \mu- \Delta c}\bigg\rceil \frac{\exp\left(-2\frac{(\Delta \mu- \Delta c)^2 \frac{W}{\Delta c}}{ (\Delta \kappa)^2}\right)}{1-\exp\left(-2\frac{(\Delta \mu- \Delta c)^2}{ (\Delta \kappa)^2}\right)}. \label{eq:E}
	\end{align}
	By Lemma \ref{lem:Hoeffding}, 
	\begin{align}
	\mathbb{P}(F^c )&\leq \exp\left(-2\frac{((\Delta \mu- \Delta c)\lceil\frac{\lambda W}{\Delta \mu- \Delta c}\rceil-W)^2}{\frac{\lambda W}{\Delta \mu- \Delta c}  (\Delta \kappa)^2}\right) \nonumber \\ &\leq \exp\left(-2\frac{(\lambda-1)^2 W(\Delta \mu- \Delta c)}{\lambda  (\Delta \kappa)^2}\right). \label{eq:F}
	\end{align}
	By \eqref{eq:E} and \eqref{eq:F}, 
	\begin{align}
	\mathbb{P}(E_t^c \cap E_{t-1}^c \cap F) \geq &1 - 2\bigg\lceil\frac{\lambda W}{\Delta \mu- \Delta c}\bigg\rceil \frac{\exp\left(-2\frac{(\Delta \mu- \Delta c)^2 \frac{W}{\Delta c}}{ (\Delta \kappa)^2}\right)}{1-\exp\left(-2\frac{(\Delta \mu- \Delta c)^2}{ (\Delta \kappa)^2}\right)} \nonumber \\ &- \exp\left(-2\frac{(\lambda-1)^2 W(\Delta \mu- \Delta c)}{\lambda  (\Delta \kappa)^2}\right). \label{eq:EcapF}
	\end{align}
	Consider the event $G = E_t^c \cap E_{t-1}^c \cap F$ and the following three cases. 
	
	Case 1: The service is hosted in state $S_H$ during time-slot $t -\big\lceil\frac{\lambda W}{\Delta \mu- \Delta c}\big\rceil$: Conditioned on $E_{t-1}^c$, by the properties of the BLTN  policy, the service is not switched to $S_L$  in time-slots $t - \big\lceil\frac{\lambda W}{\Delta \mu- \Delta c}\big\rceil$ to $t-1$. It follows that in this case, the service is hosted in state $S_H$ during time-slot $t$.
	
	Case 2: The service is hosted in state $S_L$ during time-slot $t - \big\lceil\frac{\lambda W}{\Delta \mu- \Delta c}\big\rceil$ and the state is switched to $S_H$ in time-slot $\tilde{\tau}$ such that $t - \big\lceil\frac{\lambda W}{\Delta \mu- \Delta c}\big\rceil + 1 \leq \tilde{\tau} \leq t-2$: Conditioned on $E_{t-1}^c$, by the properties of the BLTN policy, the service is not switched to $S_L$  in time-slots $\tilde{\tau}+1$ to $t-1$. It follows that in this case, the service is hosted in state $S_H$ during time-slot $t$.
	
	Case 3: The service is hosted in state $S_L$ during time-slot $t - \big\lceil\frac{\lambda W}{\Delta \mu- \Delta c}\big\rceil$ and is not switched to $S_H$ in time-slots $t - \big\lceil\frac{\lambda W}{\Delta \mu- \Delta c}\big\rceil + 1$ to $t-2$: In this case, in time-slot $t-1$, $t_{\text{evict}} \leq t - \big\lceil\frac{\lambda W}{\Delta \mu- \Delta c}\big\rceil$. Conditioned on $F$, by the properties of the BLTN policy, condition in Step 8 in Algorithm \ref{algo:BLTN} is satisfied for $\tau = t - \big\lceil\frac{\lambda W}{\Delta \mu- \Delta c}\big\rceil$.  It follows that in this case, the decision to switch states is made in time-slot $t-1$ and therefore, the service is hosted in state $S_H$ during time-slot $t$.
	
	We thus conclude that conditioned on $G = E_{t-1}^c \cap F$, the service is hosted in state $S_H$ during time-slot $t$. In addition, conditioned on $E^c_t$, the service is not switched in time-slot $t$.
	We now compute the expected cost incurred by the BLTN policy. By definition,
	$
	\mathbb{E}[C_t^{\text{BLTN}}] 
	=  \mathbb{E}[C_t^{\text{BLTN}}|G] \mathbb{P}(G) + \mathbb{E}[C_t^{\text{BLTN}}|G^c] \times \mathbb{P}(G^c).
	$
	
	Note that, 
	$
	\mathbb{E}[C_t^{\text{BLTN}}|G] = c_H + \nu -\mu_H, \ \mathbb{E}[C_t^{\text{BLTN}}|G^c] \leq W + c_L + \nu - \mu_L.
	$
	Therefore,
	\begin{align}
	\mathbb{E}[C_t^{\text{BLTN}}] 
	=  & c_H + \nu-\mu_H + (W + \Delta \mu - \Delta c) \mathbb{P}(G^c) \nonumber \\
	\leq  & c_H + \nu-\mu_H + (W + \Delta \mu - \Delta c) \times \bigg( 2\bigg\lceil\frac{\lambda W}{\Delta \mu- \Delta c}\bigg\rceil\frac{\exp(-2\frac{(\Delta \mu- \Delta c)^2\frac{W}{\Delta c}}{\Delta \kappa^2})}{1-\exp(-2\frac{(\Delta \mu- \Delta c)^2}{ (\Delta \kappa)^2})} \nonumber \\ & + \exp(-2\frac{(\lambda-1)^2 W(\Delta \mu- \Delta c)}{\lambda  (\Delta \kappa)^2})\bigg). \label{eq:finalBound}
	\end{align}
	We optimize over $\lambda>1$ to get the tightest possible bound. By Lemma \ref{lemma:optimal_causal} and \eqref{eq:finalBound}, we have the result for BLTN.

	
	Next, we consider the case when $\Delta \mu < \Delta c$. We define the following events
	
	$F_{t_1,t_2}: \displaystyle\sum_{l=t_1}^{t_2} \Delta \underbar{$X$}_l \geq \displaystyle\sum_{l=t_1}^{t_2} \Delta c+W$, $F^{\tau}=\displaystyle\bigcup_{t_1=1}^{\tau} F_{t_1,\tau}$, $F_{t-1}=\displaystyle\bigcup_{\tau=t-\lceil\frac{\lambda W}{\Delta c - \Delta \mu}\rceil}^{t-1} F^{\tau}$, $F_{t}=\displaystyle\bigcup_{\tau=t-\lceil\frac{\lambda W}{\Delta c - \Delta \mu}\rceil+1}^{t} F^{\tau}$, $E: \displaystyle\sum_{l=t-\lceil\frac{\lambda W}{\Delta c - \Delta \mu}\rceil}^{t-1} \underbar{$X$}_l+ W < \displaystyle\sum_{l=t-\lceil\frac{\lambda W}{\Delta c - \Delta \mu}\rceil}^{t-1} \Delta c$.\\
	
	By Lemma \ref{lem:Hoeffding}, it follows that $\mathbb{P}(F_{t_1,t_2})\leq \exp\left(-2\frac{(\Delta c - \Delta \mu)^2 (t_2-t_1+1)}{ (\Delta \kappa)^2}\right),$ and therefore, 
	\begin{align}
	\mathbb{P}(F^{\tau}) &\leq \displaystyle \sum_{t_1=1}^{\tau-\lceil\frac{W}{\Delta \kappa- \Delta c}\rceil+1} \exp\left(-2\frac{(\Delta c - \Delta \mu)^2 (\tau-t_1+1)}{ (\Delta \kappa)^2}\right) \nonumber \\ &\leq \frac{\exp\left(-2\frac{(\Delta c - \Delta \mu)^2 \frac{W}{\Delta \kappa- \Delta c}}{(\Delta \kappa)^2}\right)}{1-\exp\left(-2\frac{(\Delta c - \Delta \mu)^2}{ (\Delta \kappa)^2}\right)}. \label{eq:F_tau}
	\end{align}
	Using \eqref{eq:F_tau} and the union bound, $\mathbb{P}(F_{t})$ and $\mathbb{P}(F_{t-1})$ are upper bounded by
	\begin{equation}
	\bigg\lceil\frac{\lambda W}{\Delta c - \Delta \mu}\bigg\rceil \frac{\exp\left(-2\frac{(\Delta c - \Delta \mu)^2 \frac{W}{\Delta \kappa- \Delta c}}{(\Delta \kappa)^2}\right)}{1-\exp\left(-2\frac{(\Delta c - \Delta \mu)^2}{ (\Delta \kappa)^2}\right)}. \label{eq:F2}
	\end{equation}
	By Lemma \ref{lem:Hoeffding},  
	\begin{align}
	\mathbb{P}(E^c) &\leq \exp\left(-2\frac{((\Delta c - \Delta \mu)\lceil\frac{\lambda W}{\Delta c - \Delta \mu}\rceil-W)^2}{\frac{\lambda W}{\Delta c - \Delta \mu}  (\Delta \kappa)^2}\right) \leq \exp\left(-2\frac{(\lambda-1)^2(\Delta c - \Delta \mu)W}{\lambda  (\Delta \kappa)^2}\right). \label{eq:E2}
	\end{align} 
	By \eqref{eq:F2} and \eqref{eq:E2},
	\begin{align}
	\mathbb{P}(F^c_t \cap F^c_{t-1} \cap E) \geq & 1- \exp\left(-2\frac{(\lambda-1)^2(\Delta c - \Delta \mu)W}{\lambda  (\Delta \kappa)^2}\right) \nonumber \\ &-2\bigg\lceil\frac{\lambda W}{\Delta c - \Delta \mu}\bigg\rceil \frac{\exp\left(-2\frac{(\Delta c - \Delta \mu)^2 \frac{W}{\Delta \kappa- \Delta c}}{ (\Delta \kappa)^2}\right)}{1-\exp\left(-2\frac{(\Delta c - \Delta \mu)^2}{ (\Delta \kappa)^2}\right)}.
    \label{eq:FcapE2}
	\end{align}
	Consider the event $G = F^c_t \cap F^c_{t-1} \cap E$ and the following three cases. 
	
	Case 1: The service is hosted in state $S_L$ during time-slot $t -\big\lceil\frac{\lambda W}{\Delta c - \Delta \mu}\big\rceil$: Conditioned on $F^c$, by the properties of the BLTN  policy, the service is not switched to $S_H$ in time-slots $t - \big\lceil\frac{\lambda W}{\Delta c - \Delta \mu}\big\rceil$ to $t-1$. It follows that in this case, the service is hosted in state $S_L$ during time-slot $t$.
	
	Case 2: The service is hosted in state $S_H$ during time-slot $t - \big\lceil\frac{\lambda W}{\Delta c - \Delta \mu}\big\rceil$ and the state is switched to $S_L$ in time-slot $\tilde{\tau}$ such that $t - \big\lceil\frac{\lambda W}{\Delta c - \Delta \mu}\big\rceil + 1 \leq \tilde{\tau} \leq t-2$: Conditioned on $F^c$, by the properties of the BLTN policy, the state is not switched to $S_H$ in time-slots $\tilde{\tau}+1$ to $t-1$. It follows that in this case, the service is hosted in state $S_L$ during time-slot $t$.
	
	
	Case 3: The service is hosted in state $S_H$ during time-slot $t - \big\lceil\frac{\lambda W}{\Delta c - \Delta \mu}\big\rceil$ and the state is not switched to $S_L$ in time-slots $t - \big\lceil\frac{\lambda W}{\Delta c - \Delta \mu}\big\rceil + 1$ to $t-2$: In this case, in time-slot $t-1$, $t_{\text{evict}} \leq t - \big\lceil\frac{\lambda W}{\Delta c - \Delta \mu}\big\rceil$. Conditioned on $E$, by the properties of the BLTN  policy, condition in Step 16 in Algorithm \ref{algo:BLTN}  is satisfied for $\tau = t - \big\lceil\frac{\lambda W}{\Delta \mu- \Delta c}\big\rceil$.  It follows that in this case, the decision to switch states is made in time-slot $t-1$ and therefore, the service is hosted in state $S_L$ in time-slot $t$.
	
	We thus conclude that conditioned on $F^c_{t-1} \cap E$, the service is hosted in state $S_L$ during time-slot $t$. In addition, conditioned on $F^c_{t}$, the service is not switched to $S_H$ in time-slot $t$. 
	We now compute the expected cost incurred by the BLTN policy. By definition, 
	$
	\mathbb{E}[C_t^{\text{BLTN}}] =  \mathbb{E}[C_t^{\text{BLTN}}|G] \mathbb{P}(G) + \mathbb{E}[C_t^{\text{BLTN}}|G^c] \times \mathbb{P}(G^c).
	$
	
	Note that, 
	$
	\mathbb{E}[C_t^{\text{BLTN}}|G] = \nu - \mu_L + c_L, \  \mathbb{E}[C_t^{\text{BLTN}}|G^c] \leq c_H + \nu - \mu_H + W.
	$
	Therefore,
	\begin{align}
	\mathbb{E}[C_t^{\text{BLTN}}] &= \nu - \mu_L + c_L + (\Delta c - \Delta \mu + W) \mathbb{P}(G^c) \nonumber \\
	&\leq  \nu - \mu_L + c_L + (\Delta c - \Delta \mu + W) \times \Bigg( \exp\left(-2\frac{(\lambda-1)^2(\Delta c - \Delta \mu)W}{\lambda  (\Delta \kappa)^2}\right) \nonumber \\ & + 2\bigg\lceil\frac{\lambda W}{\Delta c - \Delta \mu}\bigg\rceil \frac{\exp\left(-2\frac{(\Delta c - \Delta \mu)^2 \frac{W}{\Delta \kappa- \Delta c}}{ (\Delta \kappa)^2}\right)}{1-\exp\left(-2\frac{(\Delta c - \Delta \mu)^2}{ \Delta \kappa)^2}\right)} \Bigg). 
 \label{eq:finalBound2}
	\end{align}
	We optimize over $\lambda>1$ to get the tightest possible bound. By Lemma \ref{lemma:optimal_causal} and \eqref{eq:finalBound2}, we have the result for BLTN.  
\end{proof}

\end{document}